\documentclass[10pt, twocolumn, twoside,final]{IEEEtran} % ieee TCNS
\usepackage[table,pdftex,hyperref,svgnames]{xcolor}  %search for svgnames colors to see what colors are alreadu defined

\usepackage[noadjust]{cite}

\usepackage{amsmath,amssymb,amsfonts}
\usepackage{algorithmic}
\usepackage{graphicx}
\usepackage{textcomp}

\usepackage{amsthm} %to use theorem environments
\usepackage{bm} %to use \boldsymbol
\usepackage{dsfont}
\usepackage{mathtools} 
\usepackage{hyperref}
\hypersetup{colorlinks=true, linkcolor=blue, breaklinks=true, urlcolor=blue, citecolor=blue}
\usepackage{doi}
\usepackage{float} %use H for figure/table

\usepackage[inline]{enumitem}
\usepackage{tikz} %tikz for graphs
\usepackage{pgfplots}
\usepackage{multirow} %to use \multicolumn, \multirow
\usepackage[font=small]{caption}
\usepackage[font=small]{subcaption} %subfigures
\usepackage{cleveref} %coloneqq placement

\newcommand{\mc}[1]{\mathcal{#1}}
\newcommand{\mb}[1]{\mathbb{#1}}

\newcommand{\be}{\begin{equation}}
\newcommand{\ee}{\end{equation}}

\DeclareSymbolFont{bbold}{U}{bbold}{m}{n}
\DeclareSymbolFontAlphabet{\mathbbold}{bbold}
\newcommand{\vect}[1]{\mathbbold{#1}}
\newcommand{\zeros}[1][]{\vect{0}_{#1}}
\newcommand{\ones}[1][]{\vect{1}_{#1}}

\usepackage{comment}
\excludecomment{mycomment}
\newcommand{\cut}[1]{}
\begin{mycomment}
	\renewcommand{\cut}[1]{#1}
\end{mycomment}
\usepackage{soul}

\usepackage{nccmath}

\DeclareMathOperator*{\argmax}{arg\,max}
\DeclareMathOperator*{\argmin}{arg\,min}
\usepackage{balance}

\definecolor{mygreen}{rgb}{0,0.6,0}
\definecolor{mygray}{rgb}{0.5,0.5,0.5}
\definecolor{mymauve}{rgb}{0.58,0,0.82}
\definecolor{Granata}{rgb}{0.64,0,0} 

\usepackage{etoolbox}
\usepackage[thinc]{esdiff}
\patchcmd{\thebibliography}{\chapter*}{\section*}{}{}
\usepackage{listings}
\DeclareSymbolFont{bbold}{U}{bbold}{m}{n}
\DeclareSymbolFontAlphabet{\mathbbold}{bbold}
\DeclarePairedDelimiter\ceil{\lceil}{\rceil}

\newcommand{\subscr}[2]{#1_{\textup{#2}}}
 
\newcommand{\setdef}[2]{\{#1
	\; | \; #2\}}

\newcommand{\map}[3]{#1: #2 \rightarrow #3}

\newcommand{\intersection}{\ensuremath{\operatorname{\cap}}}

\newcommand{\norm}[2]{\left\|{#1}\right\|_{#2}}
\newcommand{\real}{\mathbb{R}}
\newcommand{\realpositive}{\mathbb{R}_{{>0}}}
\newcommand{\realnonnegative}{\mathbb{R}_{\geq0}}

\newcommand{\pdv}[3][]{\frac{\partial^{#1}{#2}}{\partial{#3}^{#1}}}

\DeclareMathOperator{\diag}{diag}
\newcommand\oprocendsymbol{\hbox{$\triangle$}}
\newcommand\oprocend{\relax\ifmmode\else\unskip\hfill\fi\oprocendsymbol}

\ifpdf
\DeclareGraphicsExtensions{.pdf,.eps,.png,.jpg} \else
\DeclareGraphicsExtensions{.eps} \fi

\graphicspath{{./figs/}}

\let\OLDthebibliography\thebibliography
\renewcommand\thebibliography[1]{
	\OLDthebibliography{#1}
	\setlength{\parskip}{0pt}
	\setlength{\itemsep}{0pt plus 0.3ex}
}

\usepackage[prependcaption,colorinlistoftodos]{todonotes}

\newcommand{\until}[1]{\{1,\dots, #1\}}
\newcommand{\vleft}{\subscr{v}{left}}
\newtheorem{theorem}{Theorem}
\newtheorem{assumption}[theorem]{Assumption}

\newtheorem{lemma}[theorem]{Lemma}

\newtheorem{definition}[theorem]{Definition}

\newcommand{\mcA}{\mathcal{A}}
\pgfplotsset{compat=1.14}

\newcommand{\modelone}{the 
ASAP~\eqref{eq:dynamics} with donor-controlled 
work flow~\eqref{eq:work-eigvec}}

\makeatletter
\newcommand{\transpose}{\top}
\newcommand{\workOpt}{\bm{w}^{\mathrm{opt}}}
\newcommand{\workOpti}{w^{\mathrm{opt}}}

\renewcommand{\theenumi}{(\roman{enumi})}

\begin{document}
%\title{Co-evolving Appraisal Networks and Task Assignment: Convergence Results and Conjectures}
\title{Assign and Appraise: Achieving Optimal Performance in Collaborative Teams}
%Co-evolving Appraisal Networks and Task Assignment Dynamics
\author{Elizabeth Y. Huang, Dario Paccagnan, Wenjun Mei, and Francesco Bullo, \IEEEmembership{Fellow, IEEE}
	\thanks{Submitted on \today. This work was supported by the U.S. Army Research Laboratory, the U.S. Army Research Office under grant number W911NF-15-1-0577, and the Swiss National Science Foundation under grant number P2EZP2-181618.}
	\thanks{E.Y.H., D.P., and F.B. are with the Center for Control, Dynamical Systems and Computation, UC Santa Barbara, Santa Barbara, CA 93106-5070 USA (email: \{eyhuang, dariop, bullo\}@ucsb.edu).}
	\thanks{W.M. is with the Automatic Control Laboratory, ETH, 8092 Zurich, Switzerland (e-mail: meiwenjunbd@gmail.com).}
	}

\maketitle
% this puts page numbers in document
%\thispagestyle{plain} 
%\pagestyle{plain}

\begin{abstract}
Tackling complex team problems requires understanding each team member's skills in order to devise a task assignment maximizing the team performance. This paper proposes a novel quantitative model describing the decentralized process by which individuals in a team learn
who has what abilities, while concurrently assigning tasks
to each of the team members. In the model, the appraisal network
represents team member's evaluations of one another and each team member
chooses their own workload. The appraisals and workload assignment change simultaneously: each member builds their own local appraisal of
neighboring members based on the performance exhibited on previous tasks, while the workload is redistributed based on the current appraisal estimates. We show that the appraisal states can be reduced to a lower dimension due to the presence of conserved quantities associated to the cycles of the appraisal network. Building on this, we provide rigorous results characterizing the ability, or inability, of the team to learn each other's skill and thus converge to an allocation maximizing the team performance. We complement our analysis with extensive numerical experiments.
\end{abstract}
%\begin{abstract}
%  Tackling complex team problems requires understanding each team member's
%  skills in order to devise a task assignment maximizing the team
%  performance.  In this paper, we propose a novel quantitative model
%  describing the decentralized process by which individuals in a team learn
%  who has what skills and knowledge, while tasks are concurrently assigned
%  to each of the team members.  In the model, the appraisal network
%  represents team member's evaluations of one another and each team member
%  chooses their own workload. The appraisals and workload assignment change
%  simultaneously: each member builds their own local appraisal of
%  neighboring members based on the performance exhibited on previous tasks,
%  while the workload is redistributed based on the current appraisal
%  estimates.  We show that the appraisal states can be reduced to a lower
%  dimension due to the presence of conserved quantities associated to the
%  cycles of the appraisal network. Building on this, we provide rigorous
%  results characterizing the ability, or inability, of the team to learn
%  each other's skill and thus converge to an allocation maximizing the team
%  performance. We complement our analysis with extensive numerical
%  experiments.
%\end{abstract}
%


\section{Introduction}
Research, technology, and innovation is increasingly reliant on teams of individuals with various specializations and interdisciplinary skill sets.  In its simplest form, a group of individuals completing routine tasks is a resource allocation problem. However, tackling complex problems such as scientific research~\cite{SWJK-WJS-BSB:19}, software development~\cite{SR-RVO:13}, or problem solving~\cite{CH:12} requires consideration of the team structure, cognitive affects, and interdependencies between team members~\cite{ACG-SMF-SG-JA-PWF-FWH:18}.
In these complex scenarios, it is fundamental to discover what skills each member is endowed with, so as to devise a task assignment that maximizes the resulting collective team performance.

\subsection{Problem description}
%In this paper, we focus on a quantitative model describing the process by which individuals in a team learn who has what skills and knowledge, while tasks are concurrently assigned to each of the team members (see Figure~\ref{fig:assignappraise}). More specifically, we assume each team member is endowed with a skill level (a-priori unknown), and that the team needs to divide a complex task among its members. We let each team member build their own local appraisal of neighboring team members' based on the performance exhibited on previous tasks. Upcoming tasks are then distributed according to the current appraisal estimates. Finally, the performance of each member is newly observed by neighboring members, who, in turn, update their appraisal. Any such model satisfying these assumptions is composed of two building blocks: i) an appraisal component modeling how team members update their appraisals (bottom block in Figure~\ref{fig:assignappraise}), and ii) a work assignment component describing how the task is divided within the team (top left block in Figure~\ref{fig:assignappraise}).

%\input{figure0.tex}

%In this paper, we focus on a quantitative model describing the process by which individuals in a team learn who has what skills and knowledge, while tasks are concurrently assigned to each of the team members (see Figure~\ref{fig:assignappraise}). 
In this paper, we focus on a quantitative model describing the process by which individuals in a team evaluate one another while concurrently assigning work to each of the team members, in order to maximize the collective team performance (see Figure~\ref{fig:assignappraise}). 
More specifically, we assume each team member is endowed with a skill level (a-priori unknown), and that the team needs to divide a complex task among its members. We let each team member build their own local appraisal of neighboring team members' based on the performance exhibited on previous tasks. Upcoming tasks are then distributed according to the current appraisal estimates. Finally, the performance of each member is newly observed by neighboring members, who, in turn, update their appraisal. Any model satisfying these assumptions is composed of two building blocks: i) an appraisal component modeling how team members update their appraisals (left block in Figure~\ref{fig:assignappraise}), and ii) a work assignment component describing how the task is divided within the team (right block in Figure~\ref{fig:assignappraise}).
\begin{figure}[t]
	\centering
	\includegraphics[width=0.90\linewidth]{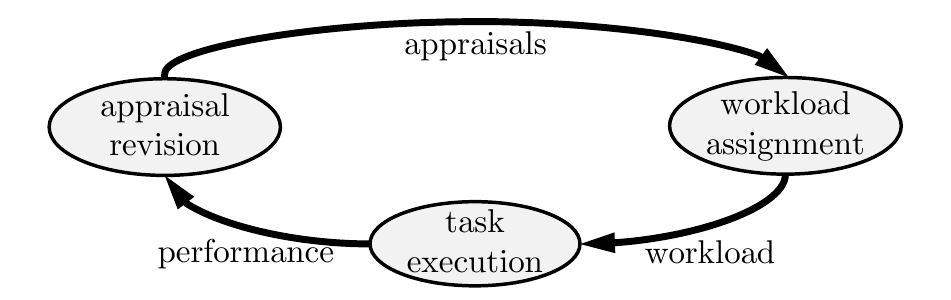}
%	\vspace*{2mm}
	\caption{Architectural overview on the assign and appraise model studied in this manuscript. Given a complex task to complete, team members get assigned and execute an initial workload (right and bottom blocks). Each team member revises their appraisal of neighboring members based on each neighbor's individual performance (left), which in turn is used to reassign the workload. The objective is for the team to learn who has what skill, so as to assign tasks in a way that maximizes the collective team performance.}
%	\vspace*{-5mm}
	\label{fig:assignappraise}
\end{figure}

We model the appraisal process i) through the lens of transactive memory systems, a conceptual model introduced by Wegner~\cite{DMW:87}, which assumes that a team is capable of developing collective knowledge regarding who has what information and capabilities. Our choice of dynamics describing the evolution of the interpersonal appraisals is inspired from replicator dynamics, whereby each team member $i$ updates their appraisal of a neighboring member $j$ proportionally to the difference between member $j$ performance and the (appraisal-weighted) average performance of the team.

%The transactive memory system model we develop is derived from replicator dynamics, and the emergence of TMS allows the team to maximize their collective performance.

%This network memory is modeled by appraisal networks, which evolves based on individual performance for the assigned task. Then the appraisal network affects how team members redistribute work among themselves. 
We model the work assignment process ii) as a compartmental system~\cite{JAJ-CPS:93}, and utilize two natural dynamics to describe how the task is divided based on the current appraisals. These dynamics correspond to utilizing different centrality measures to subdivide a complex task. It is crucial to observe that the coupling between the appraisal revision and the work assignment process results in a \emph{coevolutionary} network problem.

This paper follows a trend initiated recently, whereby many traditionally qualitative fields such as social psychology and organizational sciences are developing  \emph{quantitative} models. 
In this regard, our aim is to quantify the development of transactive memory within a team and study what conditions cause a team to fail or succeed at allocating a task optimally among members. To do so, we leverage control theoretical tools as well as ideas from evolutionary game theory, and notions from graph theory.

\subsection{Contributions}
Our main contributions are as follows.
\begin{enumerate}[wide, labelwidth=!, nolistsep]
\item We formulate a quantitative model to capture the coevolution of the workload division and appraisal network, where the optimal workload assignment maximizing the collective team performance is an equilibrium of the model. While we let the appraisal network evolve according to a replicator-like dynamics, we consider two different mechanisms for workload division and show well-posedness of the model.

\item Regardless of the mechanism used for workload division, we derive conserved quantities associated to the cycles of the appraisal network. Leveraging this result, for a team of $n$ individuals, we significantly reduce the dimension of the system from $n^2+n$ to a $2n$ dimensional submanifold. 

\item We provide rigorous positive and negative results that characterize the asymptotic behavior for either of the workload division mechanisms. When adopting the first workload division mechanism, we show that under a mild assumption, strongly connected teams are always able to learn each member's correct skill level, and thus determine the optimal workload division. In the second model variation, strong connectivity is insufficient to guarantee that the team learns the optimal workload, but more specific assumptions allow the team to converge to the optimal workload.

\item Finally, we enrich our analysis by means of numerical experiments that provide further insight into the limiting behavior. 
\end{enumerate}

\subsection{Related works}
\subsubsection*{Quantitative models of transactive memory systems}

Wegner's transactive memory systems (TMS) model~\cite{DMW:87} describes how
cognitive states affect the collective performance of a team performing
complex tasks. This widely established model captures both learning on the
individual and collective level, as well as the evolution of the
interaction between individuals within a team.

There are very few quantitative models attempting to describe TMS and most of these models rely on numerical analysis to study the evolution of team knowledge~\cite{JAG-MTB-GK-SWJK-GTC:16}, or what events are disruptive to learning and productivity in groups~\cite{EGA-KL:13}. However, numerical analysis alone has natural limitations, whereas a mathematical perspective to TMS can establish the emergence of learning behaviors for entire classes of models.
Moreover, while our proposed model is agent-based with collective knowledge represented as a weighted digraph, \cite{JAG-MTB-GK-SWJK-GTC:16,EGA-KL:13} are not agent-based models and use a scalar value to encode the team's collective knowledge.

The collective learning model introduced by Mei \emph{et al.}~\cite{WM-NEF-KL-FB:16g} was the first to quantify TMS with appraisal networks and provide convergence analysis.
In particular, for the \emph{assign/appraise} model in~\cite{WM-NEF-KL-FB:16g}, the appraisal update protocol is akin to one originally introduced in~\cite{NEF:11} and assumes each team member only updates their own appraisal based on performance comparisons. Additionally, the workload assignment is a centralized process determined by the eigenvector centrality of the network~\cite{PB:89}. 
Our model significantly differs from~\cite{WM-NEF-KL-FB:16g} in that team members update their own and neighboring team members' appraisals. Additionally, the workload assignment is a distributed and dynamic process.

%if individual $i$ observes that their own performance is higher than their perceived collective team performance, then $i$ will increase their self-appraisal and decrease the interpersonal appraisals of all their neighbors.
%{\color{red}
%\subsubsection*{Replicator dynamics}
%It is natural to apply concepts from evolutionary biology to model the learning within in a team and emergence of transactive memory. The well-known replicator dynamics from evolutionary game theory provides a model of how competing strategies propagate through a population of interacting agents~\cite{JH-KS:98,JRR-PR-MC:18}. Assuming the population is divided into $n$ subpopulations where $x_i\in[0,1]$ is the proportion of the population playing strategy $i$, then the evolution of $x_i$ is given by 
%\begin{equation}
%\label{eq:replicatorDynamics}
%  \dot{x}_i = x_i\left( f_i(x)-\bar{f}(x) \right), \qquad\text{for all } i\in\{1,\dots,n\},
%\end{equation}
%where $f_i(x)$ is the performance of subpopulation $i$ and $\bar{f}(x)=\sum_{i=1}^{n}x_if_i(x)$ is the average population fitness.
%
%Replicator dynamics also has applications outside of evolutionary biology. Both the local controller design used for engineering applications in~\cite{JBG-NQ-COM:16} and the mathematical sociology models in~\cite{WM-NEF-KL-FB:16g} stem from replicator dynamics.}

\subsubsection*{Distributed optimization}
Our model has direct ties with the field of distributed optimization.
Under suitable conditions discussed later, in fact, the team will be able
to learn each other's skill levels, and thus agree on a work assignment
maximizing the collective performance in a distributed
fashion. Additionally, any change in the problem dimension, due to the
addition or subtraction of agents, only requires local adaptions. In light
of this observation, one could reinterpret the assign and appraisal model
studied here as a distributed optimization algorithm, where the objective
is that of maximizing the team performance through local communication.  In
comparison to our work, existing distributed optimization algorithms often
require more complex dynamics. For example, \cite{AN-AO-PAP:10} requires
that the optimal solution estimates are projected back into the constrained
set, while Newton-like methods~\cite{EW-AO-AJ:13} require higher order
information.

Perhaps closest to this perspective on our problem is the work of Barreiro-Gomez \emph{et al.}~\cite{JBG-GO-NQ:16}, where evolutionary game theory is used to design distributed optimization algorithms. Nevertheless, we observe that the objective we pursue here is that of quantifying if and to what extent team members learn how to share a task optimally. In this respect, the dynamics we consider do not arise as the result of a design choice (as it is in~\cite{JBG-GO-NQ:16}), but they are rather defining the problem itself.

\subsubsection*{Adaptive coevolutionary networks}
Our model is an example of appraisal network  coevolving with a resource allocation process. Research regarding adaptive networks has gained traction in recent decades, appearing in biological systems and game theoretical applications~\cite{TG-BB:08}. Wang \emph{et al.}~\cite{ZW-MAA-ZXW-LW-CTB:15}, for example, review coupled disease-behavior dynamics, while Ogura \emph{et al.}~\cite{MO-VMP:16} propose an epidemic model where awareness causes individuals to distance themselves from infected neighbors. Finally, we note that coevolutionary game theory considers dynamics on the population strategies and dynamics of the environment, where the payoff matrix evolves with the environment state~\cite{JSW-CE-KP-SPB-WCR:16,LG-JG-MC:18}.

\subsection{Paper organization}
Section~\ref{sec:framework} contains the problem framework, model definition, the model's well-posedness, and equilibrium corresponding to the optimal workload. Section~\ref{sec:appraisal-properties} contains the properties of the appraisal dynamics and reduced order dynamics. Section~\ref{sec:convergence-eigvec} and~\ref{sec:convergence-degree} present the convergence results for the model with both workload division mechanisms.
Section~\ref{sec:fail2learn} contains numerical studies illustrating the various cases of asymptotic behavior.

\subsection{Notation}\label{sec:notation}
Let $\ones[n]$ ($\zeros[n]$ resp.) denote the $n$-dimensional column vector
with all ones (zero resp.).  Let $I_n$ represent the $n\times n$ identity
matrix.  For a matrix or vector $B\in\real^{n\times m}$, let $B\geq 0$ and $B>0$
denote component-wise inequalities. Given $x=[x_1,\dots,x_n]^\top \in
\real^{n}$, let $\diag(x)$ denote the $n\times n$ diagonal matrix such that
the $i$th entry on the diagonal equals $x_i$.  Let $\odot$ ($\oslash$
resp.) denote Hadamard entrywise multiplication (division resp.) between two
matrices of the same dimensions.
For $x,y\in\real^{n}$ and $B\in\real^{n\times n}$, we shall use the property
\begin{equation}\label{prop:Hadamard-rank1}
  xy^\transpose \odot B = \diag(x)B\diag(y).
\end{equation}
Define the $n$-dimensional simplex as $\Delta_n =
\setdef{x\in\real^{n}}{\ones[n]^\top x=1,x \geq 0}$ and the relative interior of the
simplex as $\mathrm{int}(\Delta_n) =
\setdef{x\in\mb{R}^{n}}{\ones[n]^{\top}x=1,x > 0}$.

A nonnegative matrix $B\geq 0$ is row-stochastic if $B\ones[n]=\ones[n]$.
For a nonnegative matrix $B$, $G(B)$ is the weighted digraph
associated to $B$, with node set $\{1,\dots,n\}$ and directed edge $(i,j)$
from node $i$ to $j$ if and only if $b_{ij}>0$.  A nonnegative matrix $B$
is irreducible if its associated digraph is strongly connected.  The
Laplacian matrix of a nonnegative matrix $B$ is defined as
$L(B)=\diag(B\ones[n])-B$.  For $B$ irreducible and row-stochastic,
$\subscr{v}{left}(B)$ denotes the left dominant eigenvector of $B$, i.e.,
the entry-wise positive left eigenvector normalized to have unit sum and
associated with the dominant eigenvalue of $B$~\cite[Perron Frobenius theorem]{FB:20}.

\section{Problem Framework and ASAP Model}\label{sec:framework}
In this section, we first propose the Assignment and Appraisal (ASAP) model
and establish that it is well-posed for finite time. The proposed ASAP
model can be considered a socio-inspired, distributed, and online algorithm
for optimal resource allocation problems. Our model captures two
fundamental processes within teams: workload distribution and transactive
memory. We consider two distributed, dynamic models for the workload
division: a compartmental system model and a linear model that uses
average-appraisal as the input for adjusting workload. The transactive
memory is quantified by the appraisal network and reflects individualized
peer evaluation in the team.  The development of the transactive memory
system allows the team to estimate the work assignment that maximizes the
collective team performance.
%Finally, we review an existing TMS model 
%quantified by appraisal networks. 

%\subsection{Assignment and appraisal network dynamics}
\subsection{Workload assignment, performance observation, and appraisal network}
\subsubsection*{Workload assignment}
We consider a team of $n$ individuals performing 
a sequence 
of tasks. 
Let $\bm{w} = [w_1,\dots,w_n]^\top \in 
\mathrm{int}(\Delta_n)$ denote the vector of 
\emph{workload 
assignments} for a given task, where $w_i$ is 
the work 
assignment of individual $i$. 
%We assume that every team member starts with a strictly positive workload, $\bm{w}(0)>0$.
% and Lemma~\ref{thm:finiteTimeProps} shows that the work assignment remains in the interior of the simplex for finite time.

\subsubsection*{Individual performance}
Let $\map{p(\bm w)}{\mathrm{int}(\Delta_n)}{\realpositive^{n}}$ 
represent the vector of \emph{individual performances} that change as a
function of the work assignment, where $p(\bm w) = [p_1(w_1),\dots,p_n(w_n)]^\top \in $ and $p_i(w_i)$ is the performance of
individual $i$. In general, individuals will perform better if they have
less workload; we formalize this notion with the following two
assumptions. 
\begin{assumption}(Smooth and strictly decreasing performance functions)
	\label{ass:performance-smooth}
	Assume function
	$\map{p_i}{(0,1]}{[0,\infty)}$ is $C^1$, 
	strictly decreasing,
	convex, integrable, and $\lim_{x\to 0^+}p_i(x)=+\infty$.
\end{assumption}

\begin{assumption}(Power law performance functions)
	\label{ass:performance-2}
	Assume function
	$\map{p_i}{(0,1]}{[0,\infty)}$ is of the form $p_i(x) = s_ix^{-\gamma}$ where $s_i>0$ and $\gamma\in(0,1)$.
\end{assumption}

The first assumption is quite general and can be further
weakened at the cost of additional notation. The second assumption is
more restrictive than Assumption~\ref{ass:performance-smooth}, but is
well-motivated by the power law for individual
learning~\cite{AN-PSR:81}. Note that functions obeying
Assumption~\ref{ass:performance-2} also satisfy
Assumption~\ref{ass:performance-smooth}.

\subsubsection*{Appraisal network}
Let $A=\{a_{ij}\}_{i,j \in\{1,\dots,n\}}$ 
denote the $n\times n$ nonnegative, row-stochastic 
\emph{appraisal matrix}, where 
$a_{ij}$ is individual $i$'s appraisal 
of individual $j$.  
The appraisal matrix represents the team's 
network structure and transactive memory system. 
%Since we consider team networks, it is 
%reasonable to assume the team is strongly connected and each individual appraises themself. This translates to an irreducible initial appraisal matrix $A(0)$ with strictly positive self-appraisals $a_{ii}(0)>0$ for all $i\in\{1,\dots,n\}$.

\subsection{Model description and problem statement}
In this work, we design a model where the workload assignment coevolves with the appraisals: the 
workload assignment changes as a function of 
the appraisals and the appraisals update based 
on perceived performance disparities for the assigned workload.
Suppose at each time $t$, the team has a workload assignment $\bm w(t)$, individual performances $p(\bm w(t))$, and appraisal matrix $A(t)$. Since we are studying teams, it is reasonable to assume the appraisal network is strongly connected and each individual appraises themself. This translates to an irreducible initial appraisal matrix $A(0)$ with strictly positive self-appraisals $a_{ii}(0)>0$ for all $i\in\{1,\dots,n\}$. All members also start with strictly positive workload $\bm w(0) \in \mathrm{int}(\Delta_n)$. For shorthand throughout the rest of the paper, we use $A_0 = A(0)$ and $\bm w_0 = \bm w(0)$.

Before introducing the model, first we define the work flow function $F=[F_1(A,\bm w),\dots, F_n(A,\bm w)]^\top$,
where $\map{F_i}{[0,1]^{n \times n}\times \Delta_n}{\Delta_n}$ describes how individual $i$ adjusts their own work assignment. 
Then our coevolving assignment and appraisal process is quantified by the following 
dynamical system.

\begin{definition}[ASAP (assignment and appraisal) model]
  \label{mod:assign-appraise}
  Consider $n$ performance functions $p_i$ satisfying Assumption~\ref{ass:performance-smooth} or~\ref{ass:performance-2}.
  The coevolution of the appraisal network $A(t)$  and workload assignment
  $\bm{w}(t)$ obey the following coupled
  dynamics,
  \begin{equation}
    \label{eq:dynamics}
    \begin{split}
      \dot{a}_{ij} &= a_{ij} \Big( p_j(w_j) - 
      \sum_{k=1}^{n}a_{ik}p_k(w_k) \Big),
      %	\dot{a}_{ij} &= a_{ij} \big(p_j(w_j) - 
      %	\bar{p}_i(\bm w)\big), \quad 
      %	\bar{p}_i(\bm{w})=\sum_{k=1}^{n}a_{ik}p_k(w_k)
      \\
      \dot{w}_i &= F_i(A,\bm w),
    \end{split} 
	\end{equation}
  which reads in matrix form
  \begin{equation}\label{eq:dynamics-centrality-matrixform}
    \begin{split}
      \dot{A} &= A\odot \Big( 
      \ones[n] p(\bm{w})^{\top} - 
      A p(\bm{w})\ones[n]^{\top}
      \Big), 
      \\
      \dot{\bm{w}} &= F(A,\bm w).
    \end{split}
  \end{equation}
  The work flow function $F$ obeys one of the following work flow
  models:
  \begin{align}
    \label{eq:work-eigvec}
    \text{Donor-controlled:}\quad 
    & F_i(A,\bm w) = -w_i + 
    \sum_{k=1}^{n}a_{ki}w_k,
    \\
    \label{eq:work-degree}
    \text{Average-appraisal:}\quad 
    & F_i(A,\bm w) = -w_i + \frac{1}{n}\sum_{k=1}^{n}a_{ki}.
  \end{align}
  
  The matrix forms of the donor-controlled~\eqref{eq:work-eigvec} and
  average-appraisal~\eqref{eq:work-degree} work flows are $F(A,\bm w) =
  -\bm w + A^{\top} \bm w$ and $F(A,\bm w) = -\bm w + \frac{1}{n}A^{\top}
  \ones[n]$, respectively.
\end{definition}

The appraisal weights of the ASAP model~\eqref{eq:dynamics} update based on performance feedback between neighboring individuals.
For neighboring team members $i$ and $j$, $i$ will increase their appraisal of $j$ if $j$'s performance is larger than the weighted average performance observed by $i$, i.e. $p_j(w_j)> \sum_{k=1}^{n}a_{ik}p_k(w_k)$. Individual $i$ also updates their self-appraisal with the same mechanism.
The irreducibility and strictly positive self-appraisal assumptions on the appraisal 
network means that every individual's performance is evaluated by themself and at least one other individual within the team.

The donor-controlled work 
flow~\eqref{eq:work-eigvec} 
models a team where individuals 
exchange portions of their workload assignment 
with their neighbors, and the amount of work exchanged depends on their current work assignments and the appraisal values. The work individual $j$ gives to 
individual $i$ has flow rate $a_{ji}$ and is 
proportional to $w_j$. 
The average-appraisal work flow~\eqref{eq:work-degree} assumes that each individual collects feedback from neighboring team members through appraisal evaluations. Each individual uses this feedback to calculate their average-appraisal $\frac{1}{n}\sum_{k=1}^{n}a_{ki}$, which is then used to adjust their own workload assignment. The average-appraisal is equivalent to the degree centrality of the appraisal network. Note that while the donor-controlled work flow is decentralized and distributed, the average-appraisal work flow is only distributed since it requires individuals to know the total number of team members.

%\begin{assumption}[Donor-controlled work flow]
%	\label{ass:work-eigvec}
%	Assume that workload exchange between team 
%	members is modeled by donor-controlled 
%	compartmental flow. The flow function is 
%	defined as
%	\begin{equation}
%	\label{eq:work-eigvec}
%	F_i(A,\bm w) = -w_i + \sum_{k=1}^{n}a_{ki}w_k, 
%	\end{equation}
%	which reads in matrix form as 
%	$F(A,\bm w) = -\bm w + A^{\top} \bm w$. 
%	%	\begin{equation}
%	%	F(A,\bm w) = -\bm w + A^{\top} \bm w. 
%	%	\end{equation}
%	The work flow from individual $j$ to 
%	individual $i$ has flow rate $a_{ij}$ and is 
%	proportional to $w_k$. 
%\end{assumption}
%
%\begin{assumption}[Average-appraisal work flow]
%	\label{ass:work-degree}
%	Assume that workload exchange between team 
%	members takes into account the 
%	average-appraisals of each individual, where 
%	the average-appraisal is equal to the 
%	in-degree centrality. The flow function is 
%	defined as
%	\begin{equation}
%	\label{eq:work-degree}
%	F_i(A,\bm w) = -w_i + \sum_{k=1}^{n}a_{ki}, 
%	\end{equation}
%	which reads in matrix form as 
%	$F(A,\bm w) = -\bm w + \frac{1}{n}A^{\top} 
%	\ones[n]$.
%	%	\begin{equation}
%	%	F(A,\bm w) = -\bm w + \frac{1}{n}A^{\top} 
%	%\ones[n]. 
%	%	\end{equation} 
%\end{assumption}

%\paragraph*{Assignment Dynamics}
% any given time $t\geq 0$, , where the amount of work 
%exchanged depends on their current work 
%assignment and the appraisal evaluations from 
%their in-neighbors.

%For shorthand throughout the rest of the paper, we use $\bm w(0) = \bm w_0$ and $A(0) = A_0$.
In the following lemma, we show that the ASAP 
model is well-posed and the appraisal network 
maintains the same network topology for finite 
time.

\begin{lemma}[Finite-time properties for the ASAP model]\label{thm:finiteTimeProps}
  Consider the ASAP model~\ref{eq:dynamics} with donor
  controlled~\eqref{eq:work-eigvec} or average
  appraisal~\eqref{eq:work-degree} work flow.  Assume $A_{0}$ is
  row-stochastic, irreducible, with strictly positive diagonal and $\bm
  w_0\in\mathrm{int}(\Delta_n)$.  Then for any finite $\Delta t>0$, the
  following statements hold:
  \begin{enumerate}
  	\item\label{thm:finiteTimeProps:w} $\bm w(t)\in\mathrm{int}(\Delta_n)$
  	for $t\in[0,\Delta t]$;
  	
	 	\item\label{thm:finiteTimeProps:A} $A(t)$ remains row-stochastic with the
    same zero/positive pattern for  $t\in[0,\Delta t]$.
  \end{enumerate}
\end{lemma}
\begin{proof}
	Before proving statement~\ref{thm:finiteTimeProps:w}, we give some properties of the appraisal dynamics. If $a_{ij}(t)=0$, then $\dot{a}_{ij}(t)=0$, which implies $a_{ij}(t)\geq 0$.  
	By using the Hadamard product property~\eqref{prop:Hadamard-rank1}, the matrix form of the appraisal 
	dynamics can also be written as 
	$\dot{A}=A\diag(p(\bm 
	w))-\diag(A p(\bm 
	w))A$.
	Then for $A_0\ones[n]=\ones[n]$, $\dot{A}\ones[n]=\zeros[n]$, so $A(t)$ remains row-stochastic 
	for $t\geq 0$.
	
Next, we use $A(t)$ row-stochastic to prove $w(t)\in\mathrm{int}(\Delta_n)$
for donor-controlled work flow and $t\in[0,\Delta t]$.  Left multiplying
the $\bm w(t)$ dynamics by $\ones[n]^\top$, we have $\ones[n]^\top\dot{\bm
  w}=\ones[n]^\top(-\bm w+A^\top \bm w)=\zeros[n]$. Next, let $w_i(t) =
\min_{k}\{w_k(t)\}$. For $\bm w_0\in\mathrm{int}(\Delta_n)$, $w_i(t) =
\min_{k}\{w_k(t)\} = 0$, and $A(t)\geq 0$, then $\dot w_i(t) =
\sum_{k=1}^{n}a_{ki}(t)w_k(t) \geq 0$. Therefore $\bm w(t) \in
\Delta_n$. Lastly, we apply the Gr\"onwall-Bellman Comparison Lemma to
also show that $\bm w(t)$ lives in the relative interior of the
simplex. For $w_i(0)>0$ and $\dot w_i(t) = -w_i(t) +
\sum_{k=1}^{n}a_{ki}(t)w_k(t) \geq -w_i(t)$, then $w_i(t)\geq w_i(0)
e^{-t}>0$ for $t\in[0,\Delta t]$.  Therefore, if $\bm
w_0\in\mathrm{int}(\Delta_n)$, then $\bm w(t)\in\mathrm{int}(\Delta_n)$ for
$t\in[0,\Delta t]$.
	
The proof for statement~\ref{thm:finiteTimeProps:w} can be extended to the
average-appraisal work flow~\eqref{eq:work-degree} following the same
process, since $\dot w_i(t) = -w_i(t) + \frac{1}{n}\sum_{k=1}^{n}a_{ki}(t)
\geq - w_i(t)$.
	
%	, note that there exists a sufficiently small $\tau > 0$ such that $A(t)$
%	and $w(t)$ are well defined and continuously differentiable and $p_i(w_i)-\sum_{k=1}^{n}a_{ik}p_k(w_k)$ is finite for any $i,j$ and $t\in[0,\tau]$.
%	

For statement~\ref{thm:finiteTimeProps:A}, to prove that $A(t)$ maintains
the same zero/positive pattern for $t\in[0,\Delta t]$, consider any $i,j$
such that $a_{ij}(0) > 0$.  Since $\bm w(t)\in\mathrm{int}(\Delta_n)$, then
$p(\bm w(t))> 0$ by the performance function assumptions and $p_j(w_j) -
\sum_{k=1}^{n}a_{ik}p_k(w_k)$ is finite for any $i,j$ and $t\in[0,\Delta
  t]$.  Let $p_{\max}(\bm w(t)) = \max_{k \in \{1,\dots,n\}}
\{p_{k}(w_{k})\}$. Then the convex combination of individual performances
is upper bounded by $\sum_{k=1}^{n}a_{ik}p_k(w_k) \leq p_{\max}(\bm
w(t))$. Now we can write the following lower bound for the time derivative
of $a_{ij}(t)$,
\begin{align*}
  \dot{a}_{ij}(t) &\geq 
  a_{ij}(t)\big(p_j(w_j(t)) -\sum\nolimits_{k=1}^{n}a_{ik}(t)p_k(w_k(t))\big) 
  \\
  &\geq -a_{ij}(t) p_{\max}(\bm w(t))
  .
\end{align*}

Using the Gr\"onwall-Bellman Comparison Lemma again, for
$t\in[0,\Delta t]$, then
\begin{align*}
  a_{ij}(t) &\geq a_{ij}(0)\exp \bigg( -\int_{0}^{t} p_{\max}(\bm w(\tau))d\tau \bigg) > 0.
\end{align*}
Therefore, $A(t)$ remains row-stochastic and maintains the same
zero/positive pattern as $A_0$ for finite time.
\end{proof}

\subsection{Team performance and optimal workload as model equilibria}
%We claim that in a high-performing team, each team member takes on a portion of the team's
%responsibilities relative to their own abilities, which results in all team members having equal
%performance levels.
We are interested in the collective team performance and while no single
collective team performance function is widely accepted in the social
sciences, we consider three such functions. Under minor technical
assumptions, the optimal workload for all three is characterized by equal
performance levels by the individuals and is an equilibrium point of the
ASAP model. If $p_i(w_i)$ represents the marginal utility of individual
$i$, then the collective team performance can be measured by the
\emph{total utility},
\begin{equation*}
  \subscr{\mathcal{H}}{tot}(\bm w) = \sum_{i=1}^{n} \int_{0}^{w_i} p_i(x) dx. 
\end{equation*}
The team performance can alternatively be measured by the ``weakest link'' or \emph{minimum performer},
\begin{equation*}
  \subscr{\mathcal{H}}{min}(\bm w) = \min_{i\in\until{n}}\{p_i(w_i)\}.
\end{equation*}
Another metric often used is the \emph{weighted average individual
  performance}:
\begin{equation*}
\subscr{\mathcal{H}}{avg}(\bm w) = \sum_{i=1}^{n} w_i p_i(w_i).
\end{equation*}

The next theorem clarifies when the workload maximizing either $\subscr{\mathcal{H}}{tot}$, $\subscr{\mathcal{H}}{min}$, or $\subscr{\mathcal{H}}{avg}$ is an equilibrium of the ASAP model. 
%We consider $\workOpt$ to be proportional to the the a priori unknown skill levels of the team members, with the collective team performance being maximized when the assigned workload is proportional to each member's skill.

\begin{theorem}[Optimal performance as equilibria of dynamics]
  \label{thm:equilibria}
  Consider performance functions $p_i$ satisfying
  Assumption~\ref{ass:performance-smooth} for all $i\in\{1,\dots,n\}$.  Then
  \begin{enumerate}
  \item 
  \label{thm:equilibria-pair} there exists a unique pair $(p^*,\workOpt)$ such that $p^*>0$,
    $\workOpt\in\mathrm{int}(\Delta_n)$, and $ p(\workOpt) = p^*\ones[n]$.
  \end{enumerate}
%	Consider performance functions satisfying Assumption~\ref{ass:performance-smooth}. For $\subscr{\mathcal{H}}{avg}$ and performance functions in the form $p_i(x) = \frac{s_i}{x^{\gamma}}$ for $\gamma\in(0,\infty)\notin 1$, $\subscr{\mathcal{H}}{tot}$, or $\subscr{\mathcal{H}}{min}$, then
%%

  Additionally, let $\mathcal{H}$ denote $\subscr{\mathcal{H}}{tot}$,
  $\subscr{\mathcal{H}}{min}$, or $\subscr{\mathcal{H}}{avg}$. Let Assumption~\ref{ass:performance-2} hold when $\mathcal{H}=\subscr{\mathcal{H}}{avg}$.
	%% Consider one of the following combinations of collective team
        %% performance and performance function assumptions:
        %% $\subscr{\mathcal{H}}{tot}$ with
        %% Assumption~\ref{ass:performance-smooth},
        %% $\subscr{\mathcal{H}}{min}$ with
        %% Assumption~\ref{ass:performance-smooth}, or
        %% $\subscr{\mathcal{H}}{avg}$ with
  %% Assumption~\ref{ass:performance-2}.
  Then
  \begin{enumerate}
    \setcounter{enumi}{1}
  \item\label{thm:equilibria-1} 
  $\workOpt$ is the unique solution to
    \begin{equation*}
      \label{eq:distrOptimization}
      \workOpt = \argmax_{\bm w\in\Delta_n} \{\mathcal{H}(\bm w)\}.
    \end{equation*}
  \end{enumerate}

  Finally, consider the ASAP model~\eqref{eq:dynamics} with
  donor-controlled work flow~\eqref{eq:work-eigvec} and let $A_0$ be
  row-stochastic, irreducible, with strictly positive diagonal and
  $\bm{w}_0 \in \mathrm{int}(\Delta_n)$. Then
  \begin{enumerate}
    \setcounter{enumi}{2}
  \item\label{thm:equilibria-2} there exists at least one matrix $A^*$ with
    the same zero/positive pattern as $A_0$ that satisfies
    $\workOpt=\vleft(A^*)$; and
  \item\label{thm:equilibria-3} every pair $(A^*,\workOpt)$, such that $A^*$
    has the same zero/positive pattern as $A_0$ and
    $\workOpt=\vleft(A^*)$, is an equilibrium.
  \end{enumerate}
\end{theorem}
For average-appraisal work flow~\eqref{eq:work-degree},  statements~\ref{thm:equilibria-2}-\ref{thm:equilibria-3} may not hold for $\workOpt=\frac{1}{n}(A^*)^\top \ones[n]$, since there may not exist an $A^*$ with the same zero/positive pattern as $A_0$. Section~\ref{sec:convergence-degree} elaborates on these results.

\begin{proof}
	Regarding statement~\ref{thm:equilibria-pair}, recall that $p_i$ is $C^1$ and strictly decreasing by Assumption~\ref{ass:performance-smooth} or~\ref{ass:performance-2}.
	Now we show that given our assumptions, there exists $\workOpt\in\mathrm{int}(\Delta_n)$ such that $p(\workOpt)=p^*\ones[n]$ holds. Let $p_i^{-1}$ denote the inverse of $p_i$ and let $\circ$ denote the composition of functions where $f(g(x)) = (f\circ g)(x)$. Given $p_1(w_1)=p_i(w_i)$, then $w_i = (p^{-1}_i \circ p_{1})(w_1)$ for all $i\neq 1$. Then taking into account $\workOpt\in\mathrm{int}(\Delta_n)$,
	\begin{equation*}
	w_1 + \sum\nolimits_{i=1}^{n}(p^{-1}_i \circ p_{1})(w_1) = 1.
	\end{equation*}
	$p_i$ strictly decreasing implies $p_i^{-1}$ ($p_i^{-1}\circ p_1$ resp.) is strictly decreasing (strictly increasing resp.). Therefore the left hand side of the above equation is strictly increasing, so there is a unique $\workOpti_1\in(0,1)$ solving the equation. Therefore there is a unique $(p^*,\workOpt)$ that satisfies $p(\workOpt)=p^*\ones[n]$, where $p^*=p_1(\workOpti_1)>0$.

	Regarding statement~\ref{thm:equilibria-1}, $p_i$ is strictly decreasing, $C^1$, and convex by
	Assumption~\ref{ass:performance-smooth}-\ref{ass:performance-2}. Then $\subscr{\mathcal{H}}{tot}$, $\subscr{\mathcal{H}}{min}$, and $\subscr{\mathcal{H}}{avg}$ are all strictly concave. Since we are maximizing over a compact set, and $\mathcal{H}(\bm w)$ is finite for $\bm w\in\Delta_n$, there exists a unique optimal solution
	$\workOpt\in\Delta_n$. Next we show that $\workOpt$ must satisfy $p(\workOpt)=p^*\ones[n]$ where $p^*>0$ for each collective team performance measure and $\workOpt\in\mathrm{int}(\Delta_n)$. 
	
	First, consider $\mathcal{H}=\subscr{\mathcal{H}}{tot}$. Let $\bm\mu \in\real^{n}$ and $\lambda\in\real$. Then the KKT conditions are given by:
	$ p(\workOpt) + \bm{\mu} - \lambda\ones[n] =\zeros[n] $,
	$\bm{\mu}\odot \workOpt = \zeros[n]$, and $\bm{\mu} \succeq \zeros[n]$.  If $\lambda\to \infty$, then $\workOpt=\zeros[n]$ for the first KKT condition to hold, but we require $\workOpt\in\Delta_n$. Similarly, $\workOpti_i=0$ for any $i$ would satisfy the second KKT condition, but violate the first KKT condition. As a result, $\lambda < \infty$ and $\bm\mu =
	\zeros[n]$. This implies that $p_i(\workOpti_i)=\lambda $ for all $i$. Therefore $\workOpt\in\mathrm{int}(\Delta_n)$ and there exists $p^*=\lambda \in (0,\infty) $ such that $p(\workOpt)
	= p^*\ones[n]$.
	
	Second, consider $\mathcal{H}=\subscr{\mathcal{H}}{min}$. Define the set $\argmin(p(\bm w))=\setdef{i\in\{1,\dots,n\}}{p_i(w_i)=\min_{k}\{p_k(w_k)\}}$ and let $|\argmin(p(\bm w))|$ denote the number of elements in $\argmin(p(\bm w))$. We prove the claim by contradiction. Assume $\workOpt$ is the optimal solution such that there exists at least one $j\neq i$ such that $p_i(\workOpti_i) < p_j(\workOpti_j)$ for $i\in\argmin(p(\bm w))$. Then there exists a sufficiently small $\epsilon>0$ and $\bm w^*\in\mathrm{int}(\Delta_n)$ such that $\subscr{\mathcal{H}}{min}(\workOpt) < \subscr{\mathcal{H}}{min}(\bm w^*)$, where $w^*_{i} = \workOpti_{i}-\epsilon$ and $w^*_{j} = \workOpti_{j} + \epsilon|\argmin(p(\bm w))|$.
	This contradicts the fact that $\workOpt$ is the optimal solution. Additionally, we can prove that $\workOpt\in\mathrm{int}(\Delta_n)$ by assuming there exists at least one $i$ such that $w_i=0$ and following the same proof by contradiction process. Therefore $\workOpt\in\mathrm{int}(\Delta_n)$ and $p(\workOpt)
	= p^*\ones[n]$ .
	
	Third, consider $\mathcal{H}=\subscr{\mathcal{H}}{avg}$. Let $\bm\mu \in\real^{n}$ and $\lambda\in\real$. Then the KKT conditions are given by: $ (1-\gamma)p(\bm w^*) + \bm{\mu} - \lambda\ones[n] =\zeros[n] $,
	$\bm{\mu}\odot \bm w^* = \zeros[n]$, and $\bm{\mu} \succeq \zeros[n]$. The rest of the proof follows from the same argument as used for $\mathcal{H}=\subscr{\mathcal{H}}{tot}$.
	
	Regarding statements~\ref{thm:equilibria-2} and~\ref{thm:equilibria-3},
	let $a_d = [a_{11},\dots, a_{nn}]^{\top}\in [0,1]^{n}$ and $A(a_d,A_{0}) =
	\diag(a_d) + (I_n-\diag(a_d)) A_{0}$. We prove that there exists some
	$a_d^* > 0$ such that $\workOpt = \vleft\big(A^*(a_d^*, A_{0})\big)$.
	From the assumptions on $A_{0}$, then there exists $\bar{\bm{w}} = \vleft(A_{0})$ such
	that $\sigma\bar{\bm{w}} = (I_n - \diag(a_d^*)) \workOpt$ for $\sigma \in
	\real$.  Then solving for $a_d^*$, we have $a_d^* = \ones[n] - \sigma
	(\bar{\bm{w}} \oslash \workOpt)$.  Next, we choose $\sigma =
	\epsilon/\max_{i}\{\bar{w}_i/w_i\}$ for $\epsilon \in(0,1)$, which gives
	the following bounds on $a_{ii}$ for all $i$,
	\begin{equation*}
	a_{ii} \in [1-\epsilon, 1-\epsilon \min_i\{\bar{w}_i/w_i\} (\max_i\{\bar{w}_i/w_i\})^{-1}  ] \subseteq (0,1).
	\end{equation*}
	With $a_d^* > \zeros[n]$, then $A^*(a_d^*,A_{0})$ has the same
	zero/positive pattern as $A(0)$.  This shows that, given $\workOpt$,
	there always exists a matrix $A^*$ with left dominant eigenvector
	$\workOpt$ and with the same pattern as $A(0)$.
	
	Next, we prove that any such pair $(A^*,\workOpt)$ is an equilibrium. Our
	assumptions on $A^*$ and the Perron-Frobenius theorem together imply that
	the $\mathrm{rank}(I_n-(A^*)^{\top})=n-1$.  For the ASAP
	model~\eqref{eq:dynamics} with donor-controlled work
	flow~\eqref{eq:work-eigvec}, the equilibrium conditions on the
	self-appraisal states and work assignment read:
	\begin{align} 
	\zeros[n] &= \diag\big( a_d(A^*) \big)(I_n-A^*) p(\bm{w}^*),  \label{eq:fp-1}\\
	\zeros[n] &=(A^* - I_n)^\top \bm{w}^*. \label{eq:fp-2}
	\end{align}
	Equation~\eqref{eq:fp-1} is satisfied because we know from
	statement~\ref{thm:equilibria-1} that $ p(\workOpt) = p^*\ones[n]$.
	Equation~\eqref{eq:fp-2} is satisfied because we know $\vleft(A^*) =
	\workOpt$.  This concludes the proof of statements~\ref{thm:equilibria-2}
	and~\ref{thm:equilibria-3}.
\end{proof}

  %% Lemma~\ref{thm:finiteTimeProps} shows that the
  %% dynamics~\eqref{eq:dynamics} is continuous and maps a compact set into
  %% itself. By Brouwer's fixed point theorem, the dynamics has at least one
  %% fixed point.

%Even when the same team (i.e. same individuals, performance functions, and appraisal network structure) converges to $\workOpt$, the resulting final appraisal network may result in different appraisal weights.

The equilibria described in the above lemma also resemble an evolutionarily stable set~\cite{JH-KS:98}, which is defined as the set of strategies with the same payoff. Our proof illustrates that at least one $A^*$ always exists, but in general, there are multiple $A^*$ matrices that satisfy a particular zero/positive irreducible matrix pattern with $\workOpt = \vleft(A^*)$ with the same collective team performance. 
We will later show that, under mild conditions, this optimal
solution is an equilibrium of our dynamics with various attractivity
properties (see Section~\ref{sec:convergence-eigvec}
and~\ref{sec:convergence-degree}).

\section{Properties of Appraisal Dynamics: Conserved Quantities and Reduced Order Dynamics}
\label{sec:appraisal-properties}
In this section, we show that every cycle in the appraisal network is associated to a conserved quantity. Leveraging these conserved quantities, we reduce the appraisal dynamics to an $n-1$ dimensional submanifold. Before doing so, we introduce the notion of cycles, cycle path vectors, the cycle set, and the cycle space.
%Here we introduce additional model properties for strongly connected appraisal networks. Our investigation shows that the appraisal dynamics has conserved quantities
%associated to the cycles in the network, which reduces the appraisal states to dimension $n-1$. 
%Then we define reduced order dynamics with trajectories that map back to the ASAP model.
%
%First, we provide additional notations and definitions needed for the main result of this section.  
%For a given initial appraisal matrix $A_{0}$, let edge $e=(i,j)\in\mathcal{E}(A_0)$ denote the edge connecting $i$ to $j$ for $a_{ij}(0)>0$. 
%Let $m$ denote the total number of strictly positive interpersonal appraisals in the edge set $\mathcal{E}(A_0)$. 
For a given initial appraisal matrix $A_{0}$ with strictly positive diagonal, let $m$ denote the total number of strictly positive interpersonal appraisals in the edge set $\mathcal{E}(A_0)$.
Recall that if $a_{ij}(0) = 0$ for any $i,j$, then $\dot a_{ij} = 0$, which implies $a_{ij}(t)=0$ for all $t\geq0$. 
Therefore we can consider the total number of appraisal states to be the number of edges in $A_0$, which gives a total of $n+m$ appraisal states. 

\begin{definition}[Cycles, cycle path vectors, and cycle set]
	Consider the digraph $G(A)$ associated to matrix $A\in\realnonnegative^{n\times n}$.

	A \emph{cycle} is an ordered sequence of nodes $r=\{r_1,\dots,r_{k},r_{1}\}$ with no node appearing more than once, that starts and ends at the same node, has at least two distinct nodes, and each sequential pair of nodes in the cycle denotes an edge $(r_i,r_{i+1})\in\mathcal{E}(A)$.
	We do not consider self-loops, i.e. self-appraisal edges, to be part of any cycles. 
	%	For a cycle of length $k$, we define the following convention.
	%	: if $i=k$, then $i+1=1$.
	
	Let $C_r \in \{0,1\}^{m}$ denote the \emph{cycle path vector} associated to cycle $r$. Let each off-diagonal edge of the appraisal matrix $(i,j)\in\mathcal{E}(A)$ be assigned to a number in the ordered set $\{1,\dots,m\}$. For every edge $e\in\{1,\dots,m\}$, the $e$th component of $C_r$ is defined as
	\begin{align*}
	(C_r)_e = \begin{cases}
	+1, &\;\text{if edge } e \text{ is positively traversed by } C_r,\\
	%	-1 &\;\text{if edge } e \text{ is negatively traversed by } C_r,\\
	0, &\;\text{otherwise}.
	\end{cases}
	\end{align*}
	
	Let $\Phi(A)$ denote the \emph{cycle set}, i.e. the set of all cycles, in digraph $G(A)$.
\end{definition}

To refer to a particular cycle, we will use the cycle's associated cycle path vector, which then allows us to define the cycle space. 
\begin{definition}[Cycle space]
	A \emph{cycle space} is a subspace of $\real^{m}$ spanned by cycle path vectors.
	By~\cite[pg. 29, Theorem 9]{CB:73}, the cycle space of a strongly connected digraph $G(A)$ is spanned by a basis of $\mu=m-n+1$ cycle path vectors.
	
	Let $C_B\in\{0,1\}^{m\times \mu}$ denote a matrix where the columns are a basis of the cycle space.
\end{definition}

The following theorem
\begin{enumerate*}
\item rigorously defines the conserved quantities associated to cycles in the appraisal network;
\item shows that the appraisal states can be reduced from dimension $n+m$ to $n-1$ using the conserved quantities; and
\item uses both the previous properties to introduce reduced order dynamics that have a one-to-one correspondence with the appraisal trajectories. 
\end{enumerate*}
%Our proof is inspired by the equivalence between the replicator dynamics in $n$ variables and the Lotka-Volterra dynamics in $n-1$ variables~\cite[Theorem 7.5.1]{JH-KS:98}.

\begin{theorem}[Conserved cycle constants give reduced order dynamics]
  \label{thm:reducedOrder}
  Consider the ASAP model~\eqref{mod:assign-appraise} with donor-controlled~\eqref{eq:work-eigvec} or average-appraisal~\eqref{eq:work-degree} work flow.  Given initial
  conditions $A_0$ row-stochastic, irreducible, with strictly positive
  diagonal and $\bm{w}_0\in\mathrm{int}(\Delta_n)$, let $(A(t),\bm{w}(t))$
  be the resulting trajectory. Then
  \begin{enumerate}
  \item\label{thm:reducedOrder-complete-s1} for any cycle $r$, the quantity
    \begin{equation}
      \label{eq:conservedCycleConstants}
      c_r = \prod_{ (i,j) \in r}\frac{a_{ii}(t)}{a_{ij}(t)},
    \end{equation}
    is constant; we refer to $c_r\in(0,\infty)$ as the cycle constant
    associated to cycle $r \in \Phi(A_0)$;
    
  \item\label{thm:reducedOrder-complete-s2} the appraisal matrix $A(t)$
    takes value in a submanifold of dimension $n-1$;
    
  \item\label{thm:reducedOrder-A2v} given a solution $(\bm{v}(t),
    \bar{\bm{w}}(t))\in \real_{>0}^{n} \times \mathrm{int}(\Delta_n)$ with
    initial condition $(\bm{v}_0, \bar{\bm{w}}_0) = (\vect{1}_n,\bm{w}_0)$ of
    the dynamics
	  \begin{equation} \label{eq:dynamics-reduced}
	    \begin{split}  
	      \dot{\bm v} &= \diag\big(  p(\bar{\bm w}) -
	      \bar{\bm{w}}^{\top} \mcA(\bm{v})  p(\bar{\bm{w}}) \ones[n] \big) \bm{v}, 
	      \\
	      \dot{\bar{\bm{w}}} &= F(\mcA(\bm{v}),\bar{\bm w}), 
	    \end{split}
	  \end{equation}
	  where $\mcA:\real^n\to\real^{n\times{n}}$ is defined by
	  \begin{equation}
	    \label{eq:v-to-A}
	    \begin{split}
	      \mcA(\bm{v}) &= \diag( A_0\bm{v})^{-1} A_0 \diag(\bm{v}),
	    \end{split}
	  \end{equation}
	  then $A(t)=\mcA(\bm{v}(t))$ and $\bm w(t) = \bar{\bm w}(t)$;
	  
	\item\label{thm:reducedOrder-eq} 
	for every equilibrium $(\bm{v}^*,\workOpt)$ of~\eqref{eq:dynamics-reduced},  $(A^*,\workOpt)$ is an equilibrium  of~\eqref{mod:assign-appraise} with $A^*=\mathcal{A}(\bm v^*)$;
%	the optimal equilibrium
%          point $(A^*,\workOpt)$ of~\eqref{mod:assign-appraise}
%          corresponds to the equilibrium point
%          $(\bm{v}^*,\workOpt)$
%          of~\eqref{eq:dynamics-reduced} where $\bm{v}^* >
%          0$;
	\item\label{thm:reducedOrder-complete-rank1} if additionally
          $A_0>0$, then the positive matrix $A(t)\oslash A_0$ is rank~1 for
          all time $t$.
	\end{enumerate}
\end{theorem}
%For brevity, let $\bm{v}_0=\bm v(0)$.

\begin{proof}	
	Regarding statement~\ref{thm:reducedOrder-complete-s1}, we show that
	$c_r$ is constant for any $r\in\Phi(A_0)$ by taking the natural logarithm of both
	sides of~\eqref{eq:conservedCycleConstants} and showing that the
	derivative vanishes. By Lemma~\ref{thm:finiteTimeProps},
	$\ln(c_r)$ is well-defined since $a_{ii}(t), a_{ij}(t)>0$ for any
	$a_{ij}\in r$ and finite time $t < \infty$.
	\begin{multline*}
	\diff{}{t}\ln(c_r)
	= \sum_{(i,j) \in r} 
	\Big( \frac{\dot{a}_{ii}}{a_{ii}} - 
	\frac{\dot{a}_{ij}}{a_{ij}}
	\Big)
	\\
	= \sum_{(i,j) \in r} 
	\Big( \big(p_{i}(w_i) - \bar{p}_{i}(\bm{w})\big) 
	-
	\big(p_{j}(w_{j}) - \bar{p}_{i}(\bm{w})\big) 
	\Big)
	= 0.
	\end{multline*}
	Therefore, $c_r$ is constant for all $r\in\Phi(A_0)$.

	Regarding statement~\ref{thm:reducedOrder-complete-s2}, first, we will introduce a change of variables from $A(t)$ to $B(t)=\{b_{ij}(t)\}_{i,j\in\{1,\dots,n\}}\in\realnonnegative^{n\times n}$, that comes from the  appraisal dynamics property that allows for row-stochasticity to be preserved. This allows the $n+m$ states of $A(t)$ to be reduced to $m$ states of $B(t)$. Second, we show that there exists $\mu=m-n+1$ independent cycle constants, define constraint equations associated to the cycle constants, and apply the implicit function theorem to show that the $m$ states of $B(t)$ further reduce to $n-1$ states.
	
	Let $b_{ij}(t) =
	\frac{a_{ij}(t)}{a_{ii}(t)}$ for all $i,j$. This is well-defined in finite-time by Theorem~\ref{thm:finiteTimeProps} and the assumption that $A_0$ has strictly positive diagonal. Since the diagonal entries of $B(t)$ remain constant and zero-valued edges remain zero, then we can consider the total states of $B(t)$ to be the $m$ off-diagonal edges of $B(t)$.
%	$b_{ii}=1$ and $a_{ij}=0$ implies ${b}_{ij}=0$ for all $i,j$ and $t\geq 0$, then this transformation reduces the $n+m$ appraisal states of $A$ to $m$ states of $B$.  
	Next, we introduce the cycle constant constraint functions and use the implicit function theorem to show that the $m$ states can be further reduced to $n-1$ using the cycle constants. For edge $e=(i,j)$, let $b_{ij}(t)=b_{e}(t)$.
	Let $z=[x^\top, y^\top]^\top \in \realpositive^{m}$ where $x = [b_1,\dots,b_{m-\mu}]^\top\in\realpositive^{m-\mu}$ and $y =
	[b_{m-\mu+1},\dots,b_{m}]^{\top} \in \real_{>0}^{\mu}$. 
	Consider the cycle constant constraint function
	$\map{g(x,y) = [g_1(x,y),\dots, g_{\mu}(x,y)]^\top}{\realpositive^{m-\mu}\times \realpositive^{\mu}}{\real^{\mu}}$, where $g_r(x,y) = \ln(c_r) - \sum_{(i,j)\in r}\ln(\frac{a_{ii}}{a_{ij}}) = 0$ is associated to cycle path vector $C_r$ for all $r\in\{1,\dots,\mu\}$ and the selected cycles form a basis for the cycle subspace such that $\subscr{C}{B} = [C_1,\dots,C_\mu]$. In matrix form, $g(x,y)$ reads as
	\begin{align*} 
	g(x,y)&= \begin{bmatrix}
	\ln(c_1)\\\vdots \\ \ln(c_{\mu})
	\end{bmatrix} + 
	C_B^\top \begin{bmatrix}
	\ln(b_1)\\ \vdots \\ \ln(b_{m})
	\end{bmatrix} = \zeros[\mu].
	\end{align*} 
	We partition $C_B$ into block matrices, $C_B = [\bar{C}_B^\top, \hat{C}_B^\top]^\top$ where $\bar{C}_B \in\{0,1\}^{m-\mu \times \mu}$ and $\hat{C}_B \in\{0,1\}^{\mu \times \mu}$. Then taking the partial derivative of $g(x,y)$ with respect to $y$,
	\begin{align*}
	\pdv{g(x,y)}{y} &= C_B^{\top} \begin{bmatrix}
	\zeros[m-\mu \times \mu] \\ (\diag(y))^{-1}
	\end{bmatrix} = \hat C_B^\top (\diag(y))^{-1}. 
	\end{align*} 
	The ordering of the rows of $C_B$ is determined by the ordering of the edges $e\in\{1,\dots,m\}$. Since $C_B$ is full column rank by definition, then there exists an edge ordering such that  $\mathrm{rank}(\hat C_B) = \mu$. 
	For this ordering with $\mathrm{rank}(\hat C_B) = \mu$, then $\mathrm{rank}(\pdv{g(x,y)}{y}) = \mu$. By the implicit function theorem, $y \in
	\real_{>0}^{\mu}$ is a continuous function of $x \in
	\real_{>0}^{m-\mu} = \real^{n-1}$. Equivalently, $B$ can then be reduced from $m$ states to $m-\mu = m-(m-n+1) = n-1$. Therefore if $A(t)$ is irreducible with strictly positive diagonal, then $A(t)$ can be reduced to an $n-1$ dimensional submanifold.

	Regarding statement~\ref{thm:reducedOrder-A2v}, we show that, if
        $\bm{v}(t)$ satisfies the dynamics of~\eqref{eq:dynamics-reduced},
        then $\mcA(\bm{v}(t))$ defined by equation~\eqref{eq:v-to-A}
        satisfies the original ASAP
        dynamics~\eqref{mod:assign-appraise}. For shorthand, let
        $\tilde{p}(\bm v, \bm w) = \bm{w}^{\top} \mcA(\bm{v})
         p(\bm{w})$.  We compute:
	\begin{align*}
	\dot{a}_{ij} &= \frac{a_{ij}(0) \dot{v}_{j}}{\sum\nolimits_{k=1}^{n} a_{ik}(0)v_k}  - \frac{a_{ij}(0) v_j \sum\nolimits_{k=1}^{n}a_{ik}(0)\dot{v}_{k}}{\big( \sum\nolimits_{k=1}^{n}a_{ik}(0)v_k \big)^{2}} 
	\\
	&= \frac{a_{ij}(0)v_j}{\sum\nolimits_{k=1}^{n} a_{ik}(0)v_k} \Bigg( p_j(w_j) - \tilde{p}(\bm v, \bm w)
	\\
	&\qquad - \sum_{k=1}^{n} \frac{a_{ik}(0)v_k \big( p_k(w_k) - \tilde{p}(\bm v, \bm w) \big)}{\sum\nolimits_{h=1}^{n} a_{ih}(0)v_h} \Bigg)
	\\
	&= a_{ij} \bigg( p_j(w_j) - \sum\nolimits_{k=1}^{n} a_{ik} p_k(w_k) \bigg).
	\end{align*}
  We also note that
  \begin{equation*}
  A_0 = \mcA( \bm{v}(0) ). 
  \end{equation*}
  Our claim follows from the uniqueness of solutions to ordinary differential equations.

	Statement~\ref{thm:reducedOrder-eq} follows trivially from verifying that $(\bm v^*,\workOpt)$ and $(A^*,\workOpt)$ are equilibrium points of the corresponding dynamics with $A^* = \mcA(\bm v^*) > 0$. 
%	directly from
%        Theorem~\ref{thm:equilibria}, since the equilibrium $\bm{v}^*>0$
%        implies $ p(\bm{w}^*) = (\bm{w}^*)^{\top} \mathcal{A}(\bm{v}^*)
%         p(\bm{w}^*) \ones[n]$, which only holds for $\bm{w} =
%        \workOpt$.
	
	Regarding statement~\ref{thm:reducedOrder-complete-rank1}, we first
        show that the positive matrix $A(t)\oslash A_{0}$ is rank~$1$ for
        all time $t$.  First we multiply $A(t)\oslash A_0$ by the diagonal
        matrix $D = \diag([a_{11}(0)/a_{11}, \dots,
          a_{n1}(0)/a_{n1}])$. Then we show that $D (A(t)\oslash
        A_0)$ is rank $1$, which implies that $A(t)\oslash A_0$ is also
        rank $1$.
	\begin{align*}
	D(A\oslash A_{0}) 
	&= \begin{bmatrix}
	1 &\frac{a_{11}(0)}{a_{11}} \frac{a_{12}}{a_{12}(0)} &\cdots &\frac{a_{11}(0)a_{1n}}{a_{11}a_{1n}(0)} \\
	\vdots &\vdots &\ddots &\vdots \\
	1 &\frac{a_{n1}(0)a_{n2}}{a_{n1}a_{n2}(0)} &\cdots &\frac{a_{n1}(0)a_{nn}}{a_{n1}a_{nn}(0)} 
	\end{bmatrix}
	\end{align*}
	By assumption $A_0>0$, $G(A)$ is a complete graph for finite $t$. Then the cycle constants~\eqref{eq:conservedCycleConstants}, and any nodes $i\neq j\neq k$, we have
	$\frac{a_{kk}a_{jj}a_{ii}}{a_{kj}a_{ik}a_{ji}} = \frac{a_{kk}(0)a_{jj}(0)a_{ii}(0)}{a_{kj}(0)a_{ij}(0)a_{jk}(0)}$ and $\frac{a_{kk}a_{jj}}{a_{kj}a_{jk}} = \frac{a_{kk}(0)a_{jj}(0)}{a_{kj}(0)a_{jk}(0)}$. 
	Rearranging these two equations gives $\frac{a_{ii}}{a_{ii}(0)} \frac{a_{ij}(0)}{a_{ij}} = \frac{a_{ki}}{a_{ki}(0)} \frac{a_{kj}(0)}{a_{kj}}$. This shows that every row of $D(A\oslash A_0)$ is equivalent and $\mathrm{rank}(D(A\oslash A_0)) = \mathrm{rank}(A\oslash A_0) = 1$.
\end{proof}

\subsection*{Case study for team of two}
In order to illustrate the role of the cycle
constants~\eqref{eq:conservedCycleConstants}, we consider an example of a two-person team
with performance functions $p_1(w_1) = (\frac{0.45}{w_1})^{0.9}$
and $p_2(w_2) = (\frac{0.55}{w_2})^{0.8}$.
Figure~\ref{fig:constants} shows
the evolution of the trajectories for various initial conditions of the
ASAP model with donor-controlled work flow. The trajectories
illustrate the conserved quantities associated to the cycles in the
appraisal network, which is 
\begin{equation}
\label{eq:cycle-constant-2}
c =
\frac{a_{11}(0)a_{22}(0)}{(1-a_{11}(0))(1-a_{22}(0))}
\end{equation} for the two-node
case.  Then the cycle constant $c$ with
Theorem~\ref{thm:reducedOrder}\ref{thm:reducedOrder-complete-s2}
allows us to write the dynamics for $n=2$ as
\begin{equation}\label{eq:dynamics-2node-reduced}
%\left\{
\begin{split}
\dot a_{11} &= a_{11}(1-a_{11})\big( p_1(w_1) - p_2(1-w_1) \big), \\
\dot w_1 &= -w_1 + \bigg( \frac{a_{11}(1-a_{11}) (1-c) w_1 + a_{11}}{c+a_{11}(1-c)} \bigg).
\end{split}
\end{equation}

\begin{figure}
	\centering
	\includegraphics[trim=170 285 170 
	290, clip, width = 
	\linewidth]{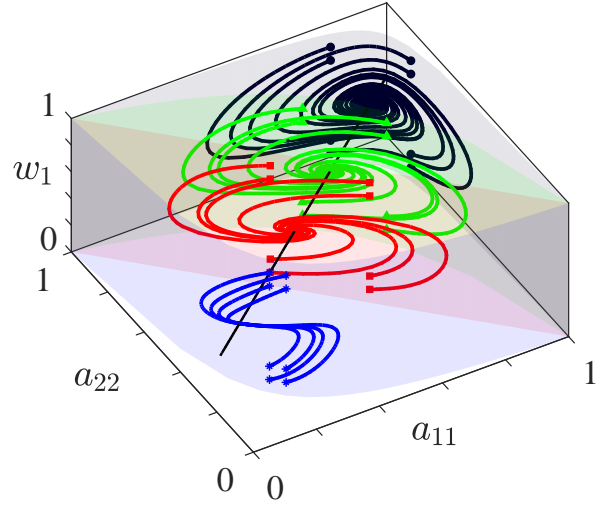}
	\caption{Trajectories of \modelone\, for 
		various initial conditions. The markers designate the initial 
		values.
		All trajectories starting on a particular 
		colored surface, remain on that colored 
		surface, where the surfaces are associated to 
		the conserved cycle constants. For $n=2$, the dynamics reduce 
		to the system~\eqref{eq:dynamics-2node-reduced} with 
		cycle constant $c$ given by~\eqref{eq:cycle-constant-2}.
		The color blue corresponds to $c<1$, red to $c=1$, and green and black to $c>1$.}
	\label{fig:constants}
\end{figure}

The cycle constants can be thought of as a parameter that measures the
level of deviation between individual's initial perception of each other's
skills. When $c_r=1$ for some $r\in\Phi(A)$, then all individuals along
cycle $r$ are in agreement over the appraisals for every other individual.
%For the complete graph with every $c_r=1$ for all $r\in\Phi(A)$, then the
%appraisal matrix is rank $1$ (see Theorem~\ref{thm:general}\ref{p2:rank1}).
%Note that in general, despite the cycle constants defined in
%Theorem~\ref{thm:reducedOrder}, appraisal weights may approach zero
%asymptotically.

\section{Stability Analysis for the ASAP Model with Donor-Controlled Work Flow}
\label{sec:convergence-eigvec}
In this section, we study the asymptotic behavior of the ASAP model with donor-controlled work flow. Our analysis is based on a Lyapunov argument. 
%This section analyzes the ASAP model~\eqref{eq:dynamics} with donor-controlled work flow~\eqref{eq:work-eigvec} for strongly connected teams.
%First, we introduce a positive definite entropy function for the ASAP dynamics.
%Second, 
Utilizing this approach, we identify initial appraisal network conditions for teams with complete graphs where the optimal workload is learned without any other additional assumptions. Under a technical assumption, we also rigorously prove that for any strongly connected team, the dynamics will converge to the optimal workload. 
%We provide numerical experiments to support the likelihood that this assumption holds.

%Before introducing the main result of this section, we provide an explanation for why $A_{0}$ must have strictly positive diagonal in order to reach $\workOpt$.
%Let $i$ be the individual such that $w^*_i = \workOpti_{\max} = \max_{k\in \{1,\dots,n\}} \{\workOpti_{k}\}$. Then
%$\workOpti_{\max} = a_{ii}\workOpti_{\max}+\sum_{k=1,k\neq i}^{n}a_{ki}\workOpti_{k} 
%\leq a_{ii}\workOpti_{\max}+\max_{k\neq i}\{a_{ki}\workOpti_{k}\},
%$
%where $\max_{k\neq i}\{a_{ki}\workOpti_{k}\} \leq \workOpti_{\max}$. If $a_{ii}=0$, then this contradicts the above constraint.  

The next lemma defines the performance-entropy function, which we show to be a Lyapunov function for the ASAP model under certain structural assumptions on the appraisal network. 

\begin{lemma}[Performance-entropy function]
  \label{thm:storageV}
  Consider the ASAP model~\eqref{mod:assign-appraise} with donor-controlled
  work flow~\eqref{eq:work-eigvec}. Assume $A_{0}$ row-stochastic,
  $\bm{w}_0 \in \mathrm{int}(\Delta_n)$, and there exists some $A^*$ with
  the same zero/positive pattern as $A_{0}$ such that $\workOpt =
  \vleft(A^*)$.  Define the \emph{performance-entropy function}
  $\map{V}{\{a_{ij}\}_{(i,j) \in \mathcal{E}(A_0)} \times
    \mathrm{int}(\Delta_n)}{\real}$ by
%  \begin{multline}
%    \label{eq:storageV}
%      V(A,\bm{w}) = \mathcal{H}(\workOpt) - \mathcal{H}(\bm w)  \\
%      -\sum_{i=1}^n
%      \workOpti_i \sum\nolimits_{\substack{k \text{ s.t. } \\
%	  (i,k)\in\mathcal{E}(A_0)}} a_{ik}^{*} \ln \Big( \frac{a_{ik}}{a_{ik}^{*}} \Big).
%  \end{multline}
%  {\color{blue} OR
        \begin{equation}
    \label{eq:storageV}
    \begin{split}
      V(A,\bm{w}) &= -
      \sum_{i=1}^{n} \bigg( \int\nolimits_{\workOpti_i}^{w_i}p_i(x) dx \\
      &\qquad
      + \workOpti_i \sum\nolimits_{\substack{k \text{ s.t.} \\
	  (i,k)\in\mathcal{E}(A_0)}} a_{ik}^{*} \ln \Big( \frac{a_{ik}}{a_{ik}^{*}} \Big) \bigg).
    \end{split}
  \end{equation}
%	}
  Then
  \begin{enumerate}
  \item $V(A,\bm{w}) > 0$ for $A\neq A^*$ or $\bm{w} \neq \workOpt$, and
  \item the Lie derivative of $V$ is 
    \begin{equation}
      \label{eq:storageV-lie}
      \dot{V}(A,\bm w) = p(\bm w)^{\transpose}(I_n-A^{\transpose}) (\bm w - \workOpt).
    \end{equation}
  \end{enumerate}
\end{lemma}
The first term of the function is the rescaled
total utility, $\subscr{\mathcal{H}}{tot}(\workOpt) - \subscr{\mathcal{H}}{tot}(\bm w) = -\sum_{i=1}^{n}
\int_{\workOpti_i}^{w_i} p_i(x) dx$.
The second term, $\workOpti_i\sum_{\substack{k \text{ s.t.} \\ (i,k)\in\mathcal{E}(A_0)}}
a_{ik}^{*} \ln \frac{a_{ik}}{a_{ik}^{*}}$, is the
Kullback-Liebler relative entropy measure~\cite{JWW:1997}.

\begin{proof}
	By Assumption~\ref{ass:performance-smooth}, $-\sum_{i=1}^{n}\int_{\workOpti_i}^{w_i} p_i(x) dx$ is convex with minimum value if and only if $w_i=\workOpti_i$. Therefore this term is positive definite for $\bm w\neq \workOpt$.
	Since the function $-\ln(\cdot)$ is strictly convex and $\sum_{k=1}^{n}a_{ik}^*=1$, Jensen's inequality can be used to give the following lower bound,
\begin{equation*}
	-\sum\nolimits_{\substack{k \text{ s.t.} \\
	(i,k)\in\mathcal{E}(A_0)}} a_{ik}^{*}\ln\Big( \frac{a_{ik}}{a_{ik}^{*}} \Big) \geq 0,
\end{equation*}
where the inequality holds strictly if and only if $A\neq A^*$. 

For the last statement of the lemma and with the assumption $\workOpt = \vleft(A^*)$, the Lie derivative of $V$ is
\begin{multline*}
\dot{V}(A,\bm w) = - p(\bm{w})^\transpose \dot{\bm{w}} - (\workOpt)^{\transpose} \big(A^* \odot (\dot{A}\oslash A) \big) \ones[n]
\\
= - p(\bm w)^{\transpose}\dot{\bm w} - (\workOpt)^{\transpose} \big(A^* \odot (\ones[n] p(\bm w)^{\transpose} - A p(\bm w)\ones[n]^\transpose) \big) \ones[n].
\end{multline*}
%Note that the Hadamard product property, $\bm{x}\bm{y}^\transpose \odot B = \diag(\bm{x})B\diag(\bm{y})$ for vectors $\bm{x},\bm{y}\in\real^{n}$ and matrix $B\in\real^{n\times n}$, gives us
Then using Hadamard product property~\eqref{prop:Hadamard-rank1} and $\workOpt = \vleft(A^*)$, $\dot{V}(A, \bm w)$
further simplifies to
%\begin{equation*}
\begin{multline*}
\dot{V}(A,\bm w) = - p(\bm w)^{\transpose}\dot{\bm w} 
\\
\quad - (\workOpt)^{\transpose} \big(A^*\diag( p(\bm{w})) - \diag(A p(\bm{w})) A^* \big) \ones[n] \\
= - p(\bm w)^{\top}\dot{\bm w} - (\workOpt)^\top \big(A^* p(\bm w) - A p(\bm w) \big) \\
= p(\bm w)^{\transpose}(I_n-A^{\transpose}) (\bm w - \workOpt). \qedhere
\end{multline*}
%\qedhere \end{equation*}
\end{proof}

The next theorem states the convergence results to the optimal workload for various cases on the connectivity of the initial appraisal matrix. For donor-controlled work flow, the optimal workload is equal to the eigenvector centrality of the network~\cite{PB:87}, which is a measure of the individual's importance as a function of the network structure and appraisal values. Therefore the equilibrium workload value quantifies each team member's contribution to the team and learning the optimal workload reflects the development of TMS within the team.
Note that statement~\ref{p3:general} relies on the assumption that conjecture given in the statement holds. This conjecture is discussed further at the end of the section, where we provide extensive simulations to illustrate its high likelihood.

\begin{theorem}[Convergence to optimal workload for strongly connected teams]
  \label{thm:general}
  Consider the ASAP model~\eqref{eq:dynamics} with donor-controlled work
  flow~\eqref{eq:work-eigvec}.  Given initial conditions $A_{0}$
  row-stochastic, irreducible, with strictly positive diagonal and $\bm
  w_0\in\mathrm{int}(\Delta_n)$.  The following statements hold:
  \begin{enumerate}
  \item\label{p1:n=2} if $n=2$ and $A_0>0$, then $\lim_{t\to\infty}\left(
    A(t),\bm w(t) \right)=(A^*,\workOpt)$ such that $A^* > 0$ is
    row-stochastic and $\workOpt=\vleft(A^*)$;
		
  \item\label{p2:rank1} if there exists $a_d(0) = [a_{11}(0),\dots a_{nn}(0)]^\top \in\mathrm{int}(\Delta_n)$ such
    that $A_{0} =  \ones[n] a_d(0)^\top$ is also rank $1$,
    then $\lim_{t\to\infty}(A(t),\bm w(t)) =
    (\ones[n](\workOpt)^\top,\workOpt)$.
    
%  \item \label{p3:general} 
%  Define $\bm{v}(t)\in\real_{>0}^{n}$ as in
%  Theorem~\ref{thm:reducedOrder}\ref{thm:reducedOrder-A2v}. 
%  If 
%%  there
%%  exists a positive constant $v_{\max}$ such that $\bm v(t) \leq v_{\max}
%%  \ones[n]$ 
%  $\bm v(t)$ is uniformly bounded
%  for all $(A_0,\bm w_0)$ and all $t\in\realnonnegative$, then
%  $\lim_{t\to\infty}( A(t),\bm{w}(t) ) = (A^*,\workOpt)$ such that $A^*$
%  is row-stochastic, has the same zero/positive pattern as $A_0$, and
%  $\workOpt = \vleft(A^*)$.
  \end{enumerate}
	Moreover, define $\bm{v}(t)\in\real_{>0}^{n}$ as in
	Theorem~\ref{thm:reducedOrder}\ref{thm:reducedOrder-A2v}. 
	\begin{enumerate}
		\setcounter{enumi}{2}
		\item 
		\label{p3:general}
		If $\bm v(t)$ is uniformly bounded for all $(A_0,\bm w_0)$ and $t\geq 0$, then
		$\lim_{t\to\infty}( A(t),\bm{w}(t) ) = (A^*,\workOpt)$ such that $A^*$
		is row-stochastic, has the same zero/positive pattern as $A_0$, and
		$\workOpt = \vleft(A^*)$.
	\end{enumerate}
\end{theorem}

\begin{proof}
	Statement~\ref{p1:n=2} follows directly from the fact that the function defined by~\eqref{eq:storageV} is a Lyapunov function for the system. For brevity, we omit the proof of Statement~\ref{p1:n=2}, since it follows a similar proof to statement~\ref{p2:rank1}.
	
	Regarding statement~\ref{p2:rank1},
	if $A_0$ is the rank $1$ form given by the theorem assumptions, then $c_r=1$ for all cycles $r\in\Phi(A_0)$ by Theorem~\ref{thm:reducedOrder}\ref{thm:reducedOrder-complete-s1}. This implies that $a_{ij}=a_{kj}$ for any $j$, all $i\neq k$, and $t\geq 0$. For the storage function $V(A,\bm w)$ as defined by~\eqref{eq:storageV}, the Lie derivative~\eqref{eq:storageV-lie} simplifies to,
	\begin{align*}
	\dot{V} &=  p(\bm w)^{\transpose}(I_n - a_d\ones[n]^\top)\bm w -  p(\bm w)^{\transpose}(I_n-a_d\ones[n]^\top) \workOpt
	\\
	&= p(\bm w)^{\transpose}(\bm w - a_d - \workOpt+a_d)  = p(\bm w)^\top (\bm w-\workOpt).
	\end{align*} 
	From $\bm w,\workOpt\in\Delta_n$, then $p(\bm{w})^{\transpose}(\bm{w}- \workOpt) = ( p(\bm{w}) -
	p(\workOpt))^{\top} (\bm{w} - \workOpt) = ( p(\bm{w}) -
	p^*\ones[n])^{\top} (\bm{w} - \workOpt)$. Since $p_i(w_i)$ strictly decreasing by Assumption~\ref{ass:performance-smooth} or~\ref{ass:performance-2}, then $\dot{V} < 0$ for $\bm w \neq \workOpt$.
	Then $V$ is a Lyapunov function for the rank $1$ initial appraisal case and $\lim_{t\to \infty}(A(t),\bm w(t)) = (\ones[n](\workOpt)^\top, \workOpt)$.
	%	 directly from the proof of the previous statement, or the storage function~\eqref{eq:storageV}

	Regarding statement~\ref{p3:general}, we start by considering the equivalent reduced order appraisal
	dynamics~\eqref{eq:dynamics-reduced} and by proving asymptotic
	convergence using LaSalle's Invariance Principle.  Define the function
	$\map{\bar{V}}{\real_{>0}^{n} \times \mathrm{int}(\Delta_n)}{\real}$, which is a modification of the storage
	function~\eqref{eq:storageV} by replacing the term $\frac{a_{ij}}{a_{ij}^*}$ with $v_i$
	for all $i,j$,
	\begin{align}
	\bar{V}(\bm v,\bm w) &= - \sum_{i=1}^{n} \bigg( \int\nolimits_{\workOpti_i}^{w_i} p_i(x)dx + \workOpti_i\ln(v_i) \bigg).
	\end{align}
	The Lie derivative of $\bar{V}$ is 
	\begin{align*}
	\dot{\bar{V}} 
	&= - p(\bm{w})^{\transpose} \dot{\bm{w}} - (\dot{\bm v}\oslash \bm v)^{\transpose} \workOpt
	\\
	&=  p(\bm{w})^{\transpose}(I_n-A^{\transpose})\bm{w} - \big(  p(\bm{w})- p(\bm w)^{\transpose} A^{\transpose} \bm{w}\ones[n] \big)^{\transpose} \workOpt
	\\
	&=  p(\bm{w})^{\transpose}(\bm{w}- \workOpt) \leq 0.
	\end{align*}
%	\begin{align*}
%	\dot{\bar{V}} &= p(\bm{w})^{\transpose}(\bm{w}- \workOpt) \leq 0.
%	\end{align*}
	We can now define the sublevel set $\Omega =
	\setdef{\bm{v}\in \real_{>0}^{n},
		\bm{w}\in\mathrm{int}(\Delta_n)}{\bar{V}( \bm{v},\bm{w} ) \leq
		\bar{V}( \bm{v}_0,\bm{w}_0 ), t\geq 0}$, which is closed and positively invariant.
	Note that if there exists any $i$ such that $\lim_{t\to \infty}v_{i}=0$, then $\lim_{t\to \infty}\bar{V}(\cdot) = \infty$. However, $\dot{\bar{V}}\leq 0$ and $\bar{V}(\bm v_0,\bm w_0)$ is finite, so $\bm v(t)$ must be bounded away from zero by a positive value for $t \geq 0$. 
	By our assumption, $\bm{v}(t)$ is also upper bounded. Then there exists constants $v_{\min},v_{\max}>0$ such that $\bm{v} \in [v_{\min}, v_{\max}]^{n}$. Then by LaSalle's Invariance Principle, the trajectories must converge to the largest invariant set contained in the intersection of 
	\begin{equation*}
	\setdef{\bm{v}\in [v_{\min}, v_{\max}]^{n},
		\bm{w}\in\mathrm{int}(\Delta_n)}{\dot{\bar{V}} = 0} \intersection \Omega.
	\end{equation*} 
	By Theorem~\ref{thm:equilibria}, if $\dot{\bar{V}} = 0$, then $w=\workOpt$ and $ p(\workOpt)=p^*\ones[n]$. This implies $\dot{\bm{v}} = \diag(\bm v)( p(\workOpt) - p^*\ones[n]) = 0$, so $\bm{v} = \bm{v}^* > 0$. 
	By Theorem~\ref{thm:reducedOrder}\ref{thm:reducedOrder-eq}, $(\bm{v}^*,\workOpt)$ corresponds to equilibrium $(A^*,\workOpt)$. Therefore $\lim_{t\to \infty}(\bm{v}(t),\bm w(t)) = (\bm{v}^*,\workOpt)$ is equivalent to $\lim_{t\to \infty}(A(t),\bm{w}(t)) = (A^*,\workOpt)$ such that $\workOpt = \vleft(A^*)$ and $A^* = \mc{A}(\bm v^*)$, where $A^*$ and $A_0$ have the same zero/positive pattern. 
\end{proof}

%We observe that assumption on the uniform-boundedness of the trajectories of $\bm v(t)$ given by the reduced order dynamics~\eqref{eq:dynamics-reduced}
%this bounded behavior holds for general initial conditions satisfying the connectivity conditions in Assumption~\ref{ass:v-boundedness}. At the end of this section, we use the 
%Monte Carlo method to support the high likelihood that the assumption holds.

Theorem~\ref{thm:general}\ref{p3:general} establishes asymptotic
convergence from all initial conditions of interest under the assumption that the trajectory $\bm v(t)$ is
uniformly bounded.  Throughout our numerical simulation studies, we have
empirically observed that this assumption has always been satisfied.  We
now present a Monte Carlo analysis~\cite{RT-GC-FD:05} to estimate the
probability that this uniform boundedness assumption holds.

For any randomly generated pair $(A_{0}, \bm w_0)$, which corresponds to
$\bm v_{0}=\ones[n]$, define the indicator function
$\map{\mathds{I}}{\realnonnegative^{n} \times
  \mathrm{int}(\Delta_n)}{\{0,1\}}$ as
\begin{enumerate}
\item $\mathds{I}(A_0,\bm{w}_0) = 1$ if there exists $v_{\max}$ such that
  $\bm v(t) \leq v_{\max}\ones[n]$ for all $t\in[0,1000]$;
	
	\item $\mathds{I}(A_0,\bm{w}_0) = 0$, otherwise.
\end{enumerate} 
Let $p=\mathds{P}[\mathds{I}( A_{0},\bm w_0 )=0]$. We estimate $p$ as
follows.  We generate $N\in\mathbb{N}$ independent identically distributed
random sample pairs, $(A_{0}^{(i)},\bm{w}_0^{(i)})$ for
$i\in\{1,\dots,N\}$, where $A_{0}^{(i)} \in [0,1]^{n\times n}$ is
row-stochastic, irreducible, with strictly positive diagonal and
$\bm{w}_0^{(i)} \in\mathrm{int}(\Delta_n)$.

Finally, we define the empirical probability as 
\[ \hat{p}_N=\frac{1}{N}\sum_{i=1}^{N}\mathds{I}( A_{0}^{(i)},\bm{w}_0^{(i)} ). \]
For any accuracy $1-\epsilon\in(0,1)$ and 
confidence level $1-\xi\in(0,1)$, then by the Chernoff Bound~\cite[Equation 9.14]{RT-GC-FD:05},
$|\hat{p}-p|<\epsilon$ 
with probability greater than confidence level $1-\xi$ if
\begin{equation}
\label{eq:chernoff-bound}
N \geq \frac{1}{2\epsilon^2}\log\frac{2}{\xi}.
\end{equation}
For $\epsilon=\xi=0.01$, the Chernoff bound~\eqref{eq:chernoff-bound} is satisfied by $N=27\,000$.

Our simulation setup is as follows. We run $27\,000$ independent MATLAB
simulations for the ASAP model~\eqref{eq:dynamics} with donor-controlled
work flow~\eqref{eq:work-degree}. We consider $n=6$, irreducible with
strictly positive diagonal $A_0$ generated using the Erd\"os-Renyi random
graph model with edge connectivity probability $0.3$, and performance
functions of the form $p_i(w_i) = (\frac{s_i}{w_i})^{\gamma_i}$ for
$\gamma_i\in(0,1)$ and $[s_1,\dots,s_n]\in\mathrm{int}(\Delta_n)$.  We find
that $\hat{p}_N=1$. Therefore, we can make the following statement.
\newline
%\begin{remark}
	\emph{
  Consider 
  \begin{enumerate*}
  	\item $n=6$;
  	\item $A_0$ irreducible with strictly positive diagonal generated by the Erd\"os-Renyi random graph model with edge connectivity probability $0.3$, and randomly generated edge weights normalized to be row-stochastic; and 
  	\item $\bm w_0\in\mathrm{int}(\Delta_n)$.
  \end{enumerate*} 
	Then with $99\%$ confidence level, there is at least
  $0.99\%$ probability that $\norm{\bm{v}(t)}{}$ is uniformly upper bounded
  for $t\in[0,1000]$.
}
%\end{remark}

%\fbtodo{This statement is way too broad: you can also talk about the
%  specific problem instances you have generated. let's talk about this
%  point}

\section{Stability Analysis for the ASAP Model with Average-Appraisal Work Flow}\label{sec:convergence-degree}

This section investigates the asymptotic behavior of the ASAP model~\eqref{eq:dynamics} with average-appraisal work flow~\eqref{eq:work-degree}.  
In contrast with the eigenvector centrality model, we observe that strongly connected teams obeying this work flow model are not always able to learn their optimal work assignment. 
First we give a necessary condition on the initial appraisal matrix and optimal work assignment for convergence to the optimal team performance.
Second, we prove that learning the optimal work assignment can be guaranteed if the team has a complete network topology or if the collective team performance is optimized by an equally distributed workload.
Note that the results in Sections~\ref{sec:framework}-\ref{sec:appraisal-properties} also hold for average-appraisal work flow, only if the equilibrium satisfies $ \workOpt = \frac{1}{n}(A^*)^\top \ones[n]$.
%Lemma~\ref{thm:finiteTimeProps} also gives well-posedness for the degree centrality model with a minor modification that is included in the body of the proof. 

%The necessary condition to have an equilibria that maximizes the collective team performance provides intuition for why average-appraisal based flow can result in lower collective team performance compared to donor-controlled work flow.
Let $\ceil{x}$ denote the ceiling function which rounds up all elements of $x$ to the nearest integer. 
The following lemma gives a condition that guarantees when the team is unable to learn the optimal workload assignment.
\begin{lemma}[Condition for failure to learn optimal work assignment for the degree centrality model]
	\label{thm:degree-centrality-necessary}
	Consider the ASAP model~\eqref{eq:dynamics} with average-appraisal work flow~\eqref{eq:work-degree}.
	Assume $A_{0}$ row-stochastic and $\bm w_0\in \mathrm{int}(\Delta_n)$.  If there exists at least one $i\in\{1,\dots,n\}$ such that 
%	$w_i(0) \leq \supscr{d}{in}(0)$, and 
	$\workOpti_i > \max\{ \frac{1}{n}\sum_{k=1}^{n}\ceil{a_{ki}(0)},w_i(0)\}$. Then $\bm w(t) \neq \workOpt$ for any $t\geq0$.
%	then there exists at least one $j\neq i$ where $a_{ji}(0)>0$ such that $\lim_{t\to\infty} a_{ji}(t) = 1$. 
%	In other words, we must have $\workOpti_i< \frac{1}{n}d_i^{\mathrm{in}}(0)$ for $\lim_{t\to\infty}\bm w(t)=\workOpt$.
\end{lemma}

\begin{proof}
  By the Gr\"onwall-Bellman Comparison Lemma, $\dot{w}_i \leq -w_i +
  \frac{1}{n}\sum_{k=1}^{n}\ceil{a_{ki}(0)}$ implies that
	\begin{align*}
	w_i(t) &\leq -w_i(0)e^{-t}+\frac{1}{n}\sum\nolimits_{k=1}^{n}\ceil{a_{ki}(0)}(e^{-t}-1) \\
	&\leq \max \Big\{ \frac{1}{n}\sum\nolimits_{k=1}^{n}\ceil{a_{ki}(0)}, w_{i}(0) \Big\}.
	\end{align*}
	Therefore if there exists at least one $i$ such that $\workOpti_i > \max\{\frac{1}{n}\sum_{k=1}^{n}\ceil{a_{ki}(0)}, w_i(0)\}$, then $w_i(t) \neq \workOpti_i$.
%	An upper bound on the average appraisal of individual $i$ is given by $\frac{1}{n}d_i^{\mathrm{in}}(0)$. Then $\dot{w}_i \leq -w_i + a_{i}^{u}$, which implies $w_i(t) \leq \frac{1}{n}d_i^{\mathrm{in}}(0)(1-e^{-t}) + w_i(0) e^{-t}$.
%	
%	Then from $w_i(t) \leq a_{i}^{u}(1-e^{-t}) + w_i(0) e^{-t}$, we have the following upper bound on the work assignment for $i\in\{1,\dots,n\}$,
%	\[
%	w_i(t)\leq \max_{t\in[0,\infty)}\{w_i(t)\}=\sum_{k=1}^{n}. 
%	\]
%	Therefore, if $\workOpti_i\geq \frac{1}{n}d_{i}^\mathrm{in}$, then $\lim_{t\to\infty}a_{ki}(t) \to 1$ for all $(k,i)\in \mathcal{E}(0)$.
\end{proof}
This sufficient condition for failure to learn the optimal workload can also be stated as a necessary condition for learning the optimal workload. In other words, if $\lim_{t\to \infty}\bm{w}(t) = \workOpt$, then $\workOpt_{i} \leq \max\{ \frac{1}{n}\sum_{k=1}^{n}\ceil{a_{ki}(0)}, w_i(0)\}$ for all $i$.

While the average-appraisal work flow does not converge to the optimal equilibrium for strongly connected teams and general initial conditions, the following lemma describes two cases that do guarantee learning of the optimal workload.

\begin{lemma}[Convergence to optimal workload for average-appraisal work flow]
	\label{thm:convergence-indegree}
	Consider the ASAP model~\eqref{eq:dynamics} with average-appraisal work flow~\eqref{eq:work-degree}. The following statements hold.
	\begin{enumerate}
		\item\label{thm:convergence-indegree:p1} If $A_{0}$ is row-stochastic, irreducible, with strictly positive diagonal, $\bm w(0)\in\mathrm{int}(\Delta_n)$, and $\workOpt =\frac{1}{n}\ones[n]$, then $\lim_{t\to\infty}(A(t),\bm w(t)) = (A^*,\frac{1}{n}\ones[n])$ where $A^*$ has the same zero/positive pattern as $A_{0}$ and is doubly-stochastic with $ \frac{1}{n}(A^*)^\transpose\ones[n] = \frac{1}{n}\ones[n]$;
		
		\item\label{thm:convergence-indegree:p2} if $A_{0}>0$ is row-stochastic and $\bm w(0)\in\mathrm{int}(\Delta_n)$,
		then $\lim_{t\to\infty}(A(t),\bm w(t)) = (A^*,\workOpt)$ where $A^*> 0$ and $\workOpt = \frac{1}{n}(A^*)^\transpose\ones[n]$.
	\end{enumerate}
\end{lemma}

\begin{proof}
	Regarding statement~\ref{thm:convergence-indegree:p1}, the storage function from~\eqref{eq:storageV} is a Lyapunov function for the given dynamics with assumption $\workOpt = \frac{1}{n}\ones[n] = \frac{1}{n}(A^*)^\top \ones[n]$. The Lie derivative $\dot{V}$ is 
	\begin{multline*}
	\dot{V}(A,\bm w)
	= - p(\bm w)^{\top}\dot{\bm w} - (\workOpt)^\top (A^*- A) p(\bm w) .
	\\
	=  p(\bm w)^{\top} \Big(\bm w - \frac{1}{n}A^\top\ones[n]- \frac{1}{n}(A^*)^\top \ones[n] + \frac{1}{n}A^\top\ones[n] \Big)  
	\\
	= p(\bm w)^{\top}(\bm w - \workOpt)
	\leq 0.
	\end{multline*}
%	Note that $V$ can be written in $(\bm v,\bm w)$ variables as 
%	\begin{align*}
%	V(\bm v,\bm w) &= \sum_{i=1}^{n}\Bigg( \int_{w^*_i}^{w_i}\big(p^*-p_i(\sigma)\big)d\sigma
%	- \frac{1}{n} \ln\frac{v_i}{\sum\limits_{i=1}^{n}a_{ik}(0)v_k} \Bigg).
%	\end{align*}
%
%	The rest of the proof follows by the same argument as the proof for Theorem~\ref{thm:convergence-eigvec-complete}. 
	By Lemma~\ref{thm:storageV}, $V = 0$ if and only if $\bm{w} = \workOpt = \frac{1}{n}\ones[n]$ and $A=A^*$ such that $\frac{1}{n}(A^*)^\top \ones[n] = \frac{1}{n}\ones[n]$. 
	Therefore $\lim_{t\to \infty}(A(t),\bm w(t)) = (A^*,\frac{1}{n}\ones[n]) $ where $A_0$ and $A^*$ have the same zero/positive pattern.
%	Since the trajectories of $(\bm{A},\bm{w})$ are bounded and $\dot{V}\leq 0$, by LaSalle's, the trajectories converge to the largest invariant set in $\setdef{(A,\bm w)\in [0,1]^{n\times n}\times \Delta_n}{\dot{V}=0} \intersection \setdef{(A,\bm w)\in [0,1]^{n\times n}\times \Delta_n}{V(\cdot) \leq V\big( A_{0},\bm w(0) \big)}$. Because $\dot{V} \leq 0$, then $a_{ij}(t)$ is bounded away from zero by a positive value for all $(i,j)\in\mathcal{E}(0)$. Therefore the largest invariant set is given by $\setdef{(A,\bm w)\in [0,1]^{n\times n}\times \Delta_n}{A=A^*, \bm w = \workOpt, \workOpt = \vleft(A^*) } $, where $A^*$ and $A_{0}$ have the same zero/positive pattern.

	Regarding statement~\ref{thm:convergence-indegree:p2}, consider the reduced order dynamics~\eqref{eq:dynamics-reduced}, with $\tilde{p}(\bm v, \bm w) = \bm{w}^{\top} \mcA(\bm{v})  p(\bm{w})$ for shorthand. 
	Define the function $\map{\bar V}{\real_{>0}^{n} \times \mathrm{int}(\Delta_n)}{\real}$ as where
	\begin{multline*}
	\bar V(\bm v,\bm w) = \sum_{i=1}^{n} \bigg(  -\int\nolimits_{\workOpti_i}^{w_i} p_i(x)dx \\
	- \workOpti_i\ln(v_i) 
	+ \frac{1}{n}\ln\Big( \sum\nolimits_{k=1}^{n} a_{ik}(0) v_{k}\Big) \bigg).
	\end{multline*}
	First, we show that $\bar V$ is lower bounded. Second, we illustrate that $\bar V$ is monotonically decreasing for $\bm w \neq \workOpt$. Then this allows us to show convergence to an optimal equilibrium.
	
	Let $a_{\min} = \min_{i,j}\{a_{ij}(0)\}$. 
	From the proof of Lemma~\ref{thm:storageV}, $-\int\nolimits_{\workOpti_i}^{w_i} p_i(x)dx \geq 0$ for all $i$. Then $\bar V$ is lower bounded by
	\begin{align*}	\label{eq:lbBarV}
	\bar{V} &\geq -\sum_{i=1}^{n} \bigg( \workOpti_i\ln(v_i) + \frac{1}{n}\ln\Big( \frac{1}{a_{\min} \norm{\bm v}{1}} \Big) \bigg)
	\\&\geq \ln(a_{\min}) - \sum_{i=1}^{n} \workOpti_i \ln\Big( \frac{v_i}{\norm{\bm v}{1}} \Big) \geq \ln(a_{\min}).
	\end{align*}

	Now we show that $\dot{\bar V} \leq 0$.
	Define the function $\map{\bm u}{\real_{>0}^{n}}{\real_{>0}^{n}}$, where $\bm u(\bm v) = \diag(A_0\bm v)^{-1}$, which reads element-wise as $u_i(\bm v) = \sum\nolimits_{k=1}^{n} a_{ik}(0) v_{k}$. Using $A(t)= \mathcal{A}(\bm v(t))$ as in~\eqref{eq:v-to-A},
%	which can be written element-wise as $a_{ij}=\frac{a_{ij}(0)v_j}{\sum_{k=1}^{n}a_{ik}(0)v_k}$, 
	then the rate of change of $\bm u$ is given by
	\begin{multline*}
		\dot{\bm u} = -\diag(\bm u)^{2} A_0\dot{\bm v} 
		= -\diag(\bm u)^{} \big(A p(\bm w)-\tilde{p}(\bm v,\bm w)\ones[n]\big).
%		\dot u_i &= \frac{-\sum_{k=1}^{n}a_{ik}(0)\dot{v}_k}{(\sum_{k=1}^{n}a_{ik}(0)v_k)^2}
%		=-u_i\sum_{k=1}^{n}a_{ik}(p_k(w_k)-\tilde p(\bm v, \bm w)).
	\end{multline*} 
%	In matrix form, this reads as $\dot{\bm u} = -\diag(\bm u) (A  p(\bm w)-\tilde p(\bm v,\bm w)\ones[n])$.
	Plugging $\bm u$ into $\bar V$, the Lie derivative of $\bar{V}$ is
	\begin{multline*}
	\dot{\bar V}(\bm v, \bm w) 
	= -  p(\bm w)^\top\dot{\bm w} - ( \dot{\bm v}\oslash\bm v)^\top \bm \workOpt - \frac{1}{n}( \dot{\bm u}\oslash\bm u)^\top\ones[n] 
	\\
	= - p(\bm w)^{\top} (-\bm w+\frac{1}{n}A^\top \ones[n]) - ( p(\bm w) - \tilde p(\bm v,\bm w)\ones[n])^\top \workOpt \\- \frac{1}{n}(-A p(\bm w)+\tilde{p}(\bm v,\bm w)\ones[n])^\top \ones[n]
	\\
	=  p(\bm w)^\top \big(\bm w-\workOpt \big) + \tilde p(\bm v,\bm w)\big( \ones[n]^\top \workOpt - \frac{1}{n} \ones[n]^\top\ones[n]\big)
	\\
	= p(\bm w)^\top (\bm w - \workOpt) 
	\leq 0.
	\end{multline*}
%	\liztodo{can this be proven w/o boundedness of v}
	Since $\dot{\bar V} \leq 0$, implies that $\bar V(\bm v,\bm w) \leq \bar V(\bm v_0,\bm w_0) < \infty$, we can conclude that there exists some strictly positive constant $v_{\min}>0$ such that $\bm v \geq v_{\min}\ones[n]$.
	
	Note that $\dot{\bar V} = 0$ if and only if $\bm w = \workOpt$ by Lemma~\ref{thm:storageV}. Because $\bar V$ has a finite lower bound and is monotonically decreasing for $\bm w \neq \workOpt$, then as $t\to\infty$, $\bar V$ will decrease to the level set where $\bm w = \workOpt$. Then $\bm w = \workOpt$ implies $\dot{\bm w} = 0$ and $\dot{\bm v} = 0$. Therefore $\lim_{t\to \infty}(\bm v,\bm w) = (\bm v^*,\workOpt)$ such that $\workOpt = \frac{1}{n}\mathcal{A}(\bm v^*)^{\top}\ones[n] = \frac{1}{n}(A^*)^\top \ones[n]$.
%	
%	Since $\lim_{v_i\to \infty}\bar{V} = -\infty$ and $\lim_{v_i\to 0} \bar{V} = \infty$, the set 
%	\begin{multline*}
%	\Omega = \{\bm v\in\real_{>0}^{n}, \bm w \in\mathrm{int}(\Delta_n) | \\ \bar V(\cdot) \in[\ln(\min_{i,j}\{a_{ij}(0)\}), \bar V(\bm v_0, \bm w_0)]\}
%	\end{multline*} is compact positively invariant. 
%	Because $\lim_{v_i\to 0} \bar{V} = \infty$, then $\dot{\bar V}\leq 0$ also implies $v_i(t)$ is bounded away from $0$ by a positive value for all $t\geq 0$.
%	Then by LaSalle's Invariance principle, the trajectories converge to the largest set in $\Omega \cap \setdef{\bm v\in\real_{>0}^{n}, \bm w \in\mathrm{int}(\Delta_n)}{\dot{\bar{V}}=0}$. 
%	The rest of the proof follows by the same argument as the proof for Theorem~\ref{thm:convergence-eigvec-complete} with trajectories asymptotically approaching the invariant set $\setdef{(A,\bm w)\in [0,1]^{n\times n}\times \Delta_n}{A=A^*>0, \bm w = \workOpt, \text{ such that } \workOpt = \vleft(A^*)} $.
\end{proof}

\section{Numerical Simulations}
\label{sec:fail2learn}
In this section, we utilize numerical simulations to investigate various cases of the ASAP model to illustrate when teams succeed and fail at optimizing their collective performance.

%\subsection{Ability to learn the optimal work assignment}

For all the simulations in this section, we consider performance functions of the form $p_i(w_i) = (\frac{s_i}{w_i})^{\gamma}$ for $\gamma\in(0,1)$ and all $i$, which satisfy Assumptions~\ref{ass:performance-smooth}-\ref{ass:performance-2}. Then the same optimal workload maximizes any choice of collective team performance we have introduced.

First, we provide an example of a team with a strongly connected appraisal network and strictly positive self-appraisal weights, i.e. satisfying the assumptions of Theorem~\ref{thm:general}\ref{p3:general}, to illustrate a case where the team learns the optimal work assignment.  Figure~\ref{fig:ex-eigvec} illustrates the evolution of the appraisal network and work assignment of the ASAP model~\eqref{eq:dynamics} with donor-controlled work flow~\eqref{eq:work-eigvec}. 
\begin{figure}[ht!]
	\centering
	\begin{tabular}{@{} r @{} p{0.95\linewidth-2\tabcolsep} @{}}
		$\workOpt$ &\begin{tabular}{c}\includegraphics[width=0.97\linewidth-2\tabcolsep]{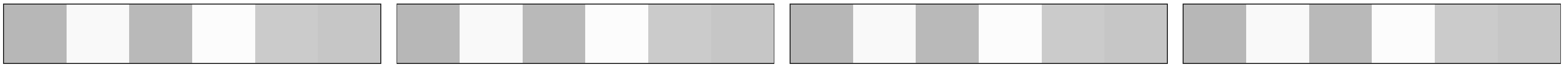}\end{tabular}
		\\
		$\bm w(t)$ &\begin{tabular}{c}\includegraphics[width=0.97\linewidth-2\tabcolsep]{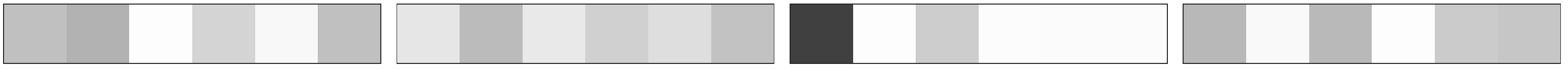}\end{tabular}
		\\
		$A(t)$ &\begin{tabular}{c}\includegraphics[width=0.97\linewidth-2\tabcolsep]{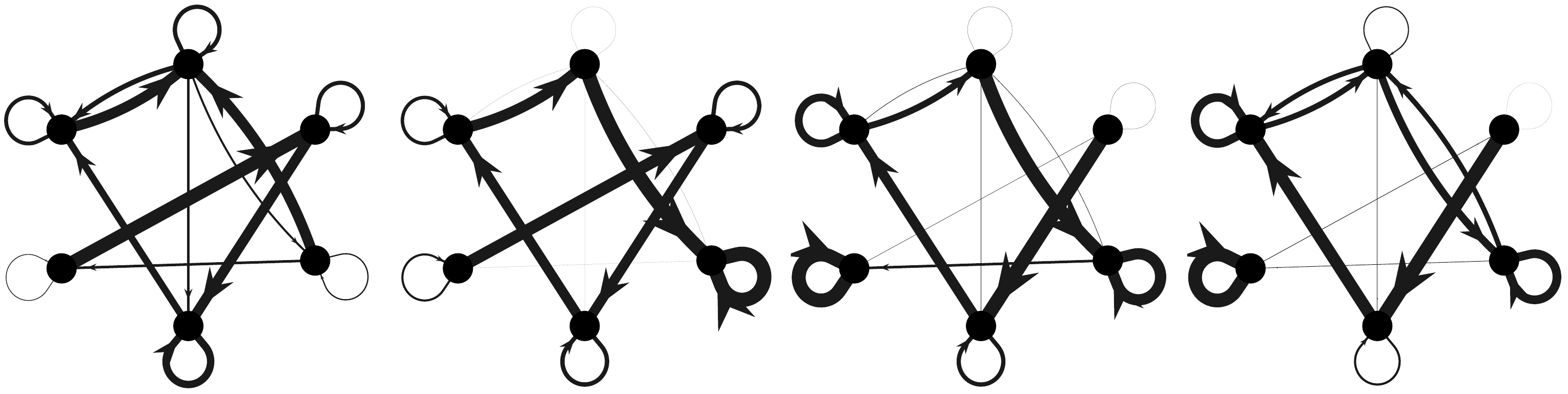}\end{tabular}
	\end{tabular}
	\caption{Visualization of the evolution of $\bm w(t)$ and $A(t)$ obeying the ASAP Model~\eqref{eq:dynamics} with donor-controlled work flow~\eqref{eq:work-eigvec}. For the work assignment vector, the darker the entry, the higher value it has. For the appraisal matrix, the thicker the edge is, the higher the appraisal edge weight is. The team's initial appraisal network is strongly connected with strictly positive self-appraisals, and is an example of a team that successfully learns the work assignment that maximizes the collective team performance.
%		 with $\workOpt = [0.28, 0.02, 0.27, 0.01, 0.20, 0.22]^{\top}$. 
%	The performance functions are of the form $p_i(w_i)=(\workOpti_i/w_i)^{\gamma}$ for $\gamma\in(0,1)$ and all $i$. 
	The plots pictured are at times  $t=\{0,1,10,1000\}$, from left to right.}
	\label{fig:ex-eigvec}
\end{figure}

\subsection{Distributed optimization illustrated with switching team members}
Next we consider another example of the ASAP model~\eqref{eq:dynamics} with donor-controlled work flow~\eqref{eq:work-eigvec}, where individuals are switching in and out of the team. Under the behavior governed by the ASAP model, only affected neighboring individuals need to be aware of an addition or subtraction of a team member, since the model is both distributed and decentralized. In this example, when individual $j$ is added to the team as a neighbor of individual $i$, $i$ allocates a portion of their work assignment to the new individual $j$. Similarly, if individual $j$ is removed, then $j$'s neighbors will absorb $j$'s workload. Let $k=1$, $k=2$, and $k=3$ denote the subteams from time intervals $t\in[0,5)$, $t\in[5,15)$, and $t\in[15,\infty)$, respectively. Then let $\subscr{\mathcal{H}}{tot}^{(k)}$ denote the collective performance for the $k$th subteam.  Figure~\ref{fig:add-subtract-node} illustrates the appraisal network topologies of each subteam and the evolution of the workload $\bm w(t)$ and normalized collective team performance $\subscr{\mathcal{H}}{tot}^{(k)}$.

\begin{figure}[h!]
	\begin{minipage}[t]{0.75\linewidth}
		\vspace{0pt}
		\includegraphics[width=\linewidth]{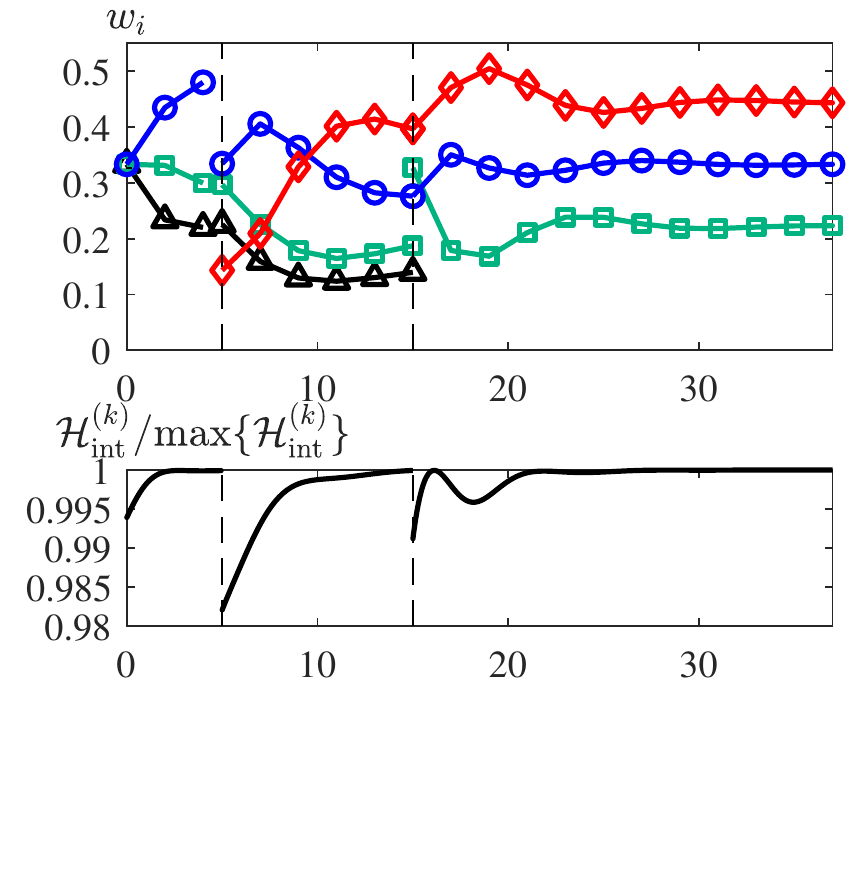}
	\end{minipage}
	\begin{minipage}[t]{0.2\linewidth}
		\vspace{4mm}
		\includegraphics[width=\linewidth]{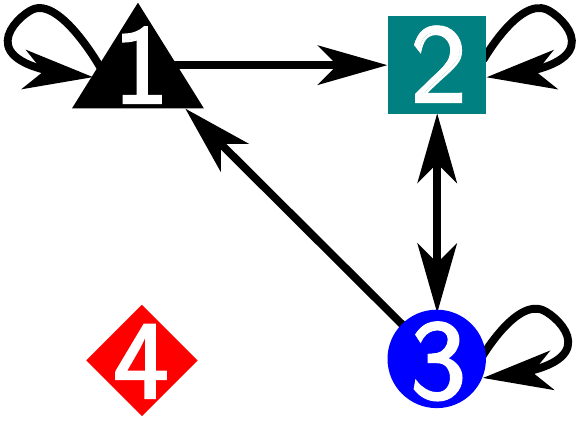}
		\\
		\vspace{0.1mm}
		\includegraphics[width=\linewidth]{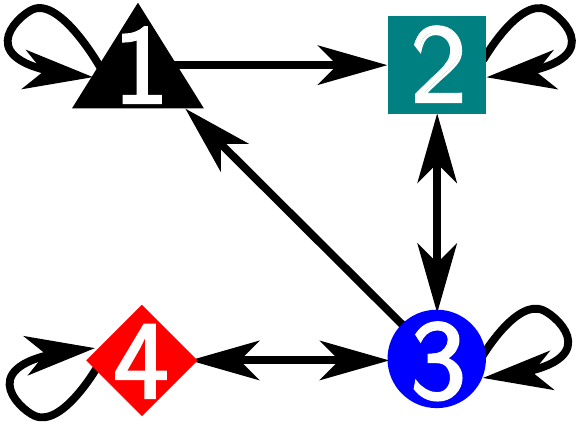}
		\\
		\vspace{0.1mm}
		\includegraphics[width=\linewidth]{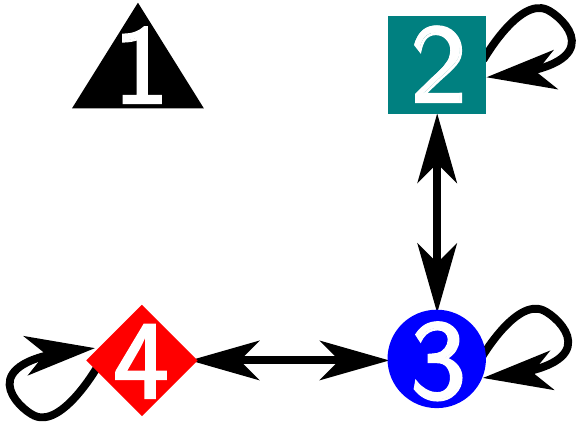}
		\vfill
	\end{minipage}
	\caption{Evolution of the ASAP model~\eqref{eq:dynamics} with donor-controlled work flow~\eqref{eq:work-eigvec} where individuals are being added and removed from the team. From top to bottom, the digraphs depict the topology of the team for $t\in[0,5)$, $t\in[5,15)$ and $t\in[15,\infty)$. At $t=10$, individual $4$ (in red diamond) is added to the team and individual $3$ gives a portion of their work to individual $4$. At $t=20$, individual $1$ (in black triangle) is removed from the team, and $1$'s work assignment is given to individual $2$. 
	}
	\label{fig:add-subtract-node}
\end{figure}

\subsection{Failure to learn}

\subsubsection*{Partial observation of performance feedback does not guarantee learning optimal work assignment}
Partial observation occurs when the appraisal network does not have the desired strongly connected property, resulting in team members having insufficient feedback to determine their optimal work assignment.
We consider an example of the ASAP model~\eqref{eq:dynamics} with donor-controlled work flow~\eqref{eq:work-eigvec} and reducible initial appraisal network $A_{0}$. 
Figure~\ref{fig:ex-fail-reducible-eigvec} illustrates how some appraisal weights between neighboring individuals approach zero asymptotically, resulting in the team not being capable of learning the work distribution that maximizes the collective team performance.

%\begin{figure}
%	\begin{tabular}{@{} r @{}@{} p{0.95\columnwidth} @{}}
%		$\bm w^*$
%		%		&\begin{tabular}{r}$$\end{tabular}
%		&\begin{tabular}{c}
%			\includegraphics[width=0.95\linewidth-3\tabcolsep]{figures/examples-trajectory-evolution/eigvec-reducible-skill.pdf}\end{tabular} 
%		\\
%		$\bm w$
%		%		&\begin{tabular}{r}$\bm w$\end{tabular}
%		&\begin{tabular}{c}
%			\includegraphics[width=0.95\linewidth-3\tabcolsep]{figures/examples-trajectory-evolution/eigvec-reducible-work.pdf}\end{tabular}
%		\\
%		$A$
%		%		&\begin{tabular}{r}$A$\end{tabular} 
%		&\begin{tabular}{c}
%			\includegraphics[width=0.95\linewidth-3\tabcolsep]{figures/examples-trajectory-evolution/eigvec-reducible-appraisal.pdf}
%		\end{tabular}
%	\end{tabular}
%	\caption{Cases of failure to learn the optimal work assignment $\workOpt = $.Example of weak connectivity being insufficient to guarantee learning and causing edge weights to asymptotically approach $0$ for the ASAP model~\eqref{eq:dynamics} with average-appraisal work flow~\eqref{eq:work-eigvec}.
%	}
%	\label{fig:example-reducible-eigvec}
%\end{figure}

\begin{figure}[h!]
%	\begin{subfigure}{\columnwidth}
%		\begin{tabular}{@{} r @{} p{0.97\linewidth-2\tabcolsep} @{}}
%			$\workOpt$ &\begin{tabular}{c}\includegraphics[width=0.97\linewidth-2\tabcolsep]{figures/examples-trajectory-evolution/eigvec-skill.pdf}\end{tabular}
%			\\
%			$\bm w$ &\begin{tabular}{c}\includegraphics[width=0.97\linewidth-2\tabcolsep]{figures/examples-trajectory-evolution/eigvec-work.pdf}\end{tabular}
%			\\
%			$A$ &\begin{tabular}{c}\includegraphics[width=0.97\linewidth-2\tabcolsep]{figures/examples-trajectory-evolution/eigvec-appraisal.pdf}\end{tabular}
%		\end{tabular}
%		\caption{Succeeding to learn the optimal work assignment: evolution of the ASAP Model~\eqref{eq:dynamics} with donor-controlled work flow~\eqref{eq:work-eigvec} and strongly connected appraisal network.}
%		\label{fig:ex-eigvec}
%	\end{subfigure}
	\begin{subfigure}{\columnwidth}
		\centering
		\begin{tabular}{@{} r @{}@{} p{0.95\linewidth-2\tabcolsep} @{}}
			$\workOpt$
			&\begin{tabular}{c}
				\includegraphics[width=0.97\linewidth-2\tabcolsep]{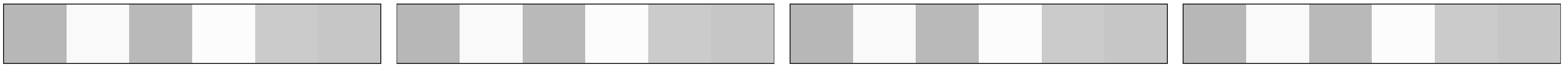}\end{tabular} 
			\\
			$\bm w(t)$
			&\begin{tabular}{c}
				\includegraphics[width=0.97\linewidth-2\tabcolsep]{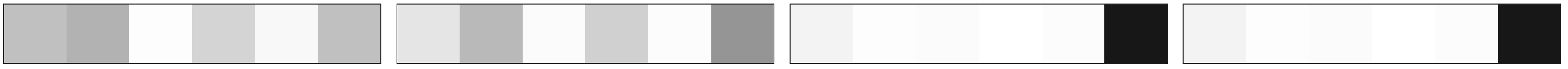}\end{tabular}
			\\
			$A(t)$
			&\begin{tabular}{c}
				\includegraphics[width=0.97\linewidth-2\tabcolsep]{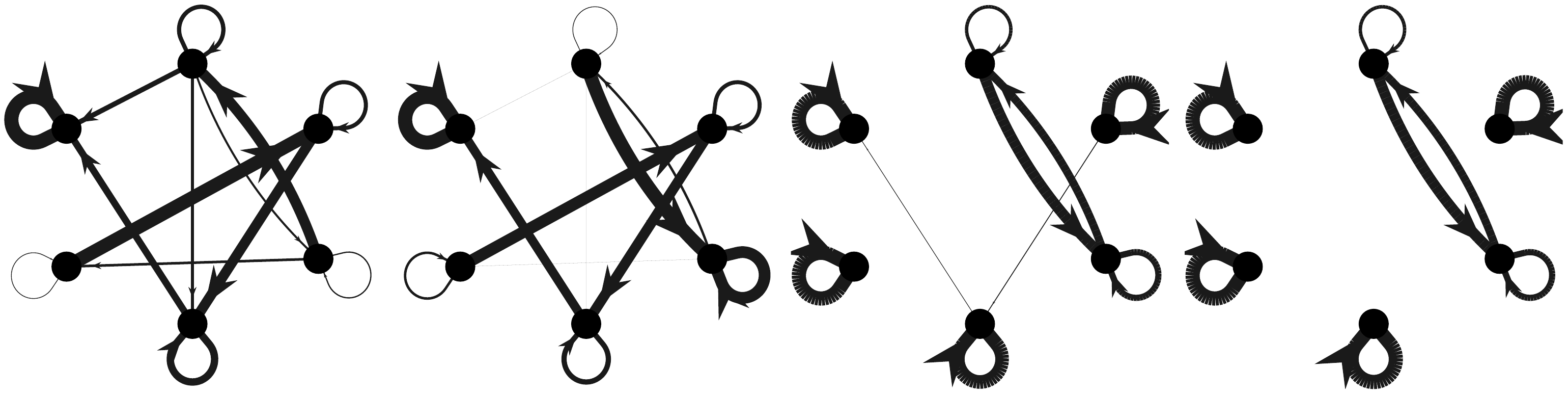}
			\end{tabular}
		\end{tabular}
		\caption{Visualization of the evolution of $\bm w(t)$ and $A(t)$ obeying the ASAP model~\eqref{eq:dynamics} with donor-controlled work flow~\eqref{eq:work-eigvec} and $A_0$ weakly connected.}
		\label{fig:ex-fail-reducible-eigvec}
	\end{subfigure}
	\begin{subfigure}{\columnwidth}
		\centering
		\begin{tabular}{@{} r @{}@{} p{0.95\linewidth-2\tabcolsep} @{}}
			$\workOpt$
			%		&\begin{tabular}{r}$\bm w^*$\end{tabular}
			&\begin{tabular}{c}
				\includegraphics[width=0.97\linewidth-2\tabcolsep]{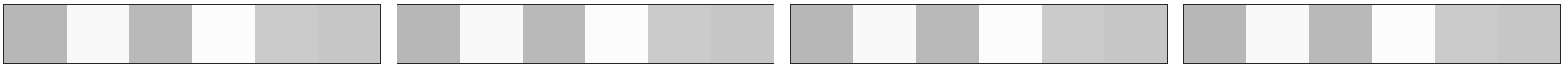}\end{tabular} 
			\\
			$\bm w(t)$
			%		&\begin{tabular}{r}$\bm w$\end{tabular}
			&\begin{tabular}{c}
				\includegraphics[width=0.97\linewidth-2\tabcolsep]{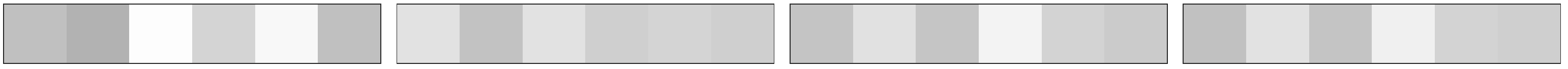}\end{tabular}
			\\
			$A(t)$
			%		&\begin{tabular}{r}$A$\end{tabular} 
			&\begin{tabular}{c}
				\includegraphics[width=0.97\linewidth-2\tabcolsep]{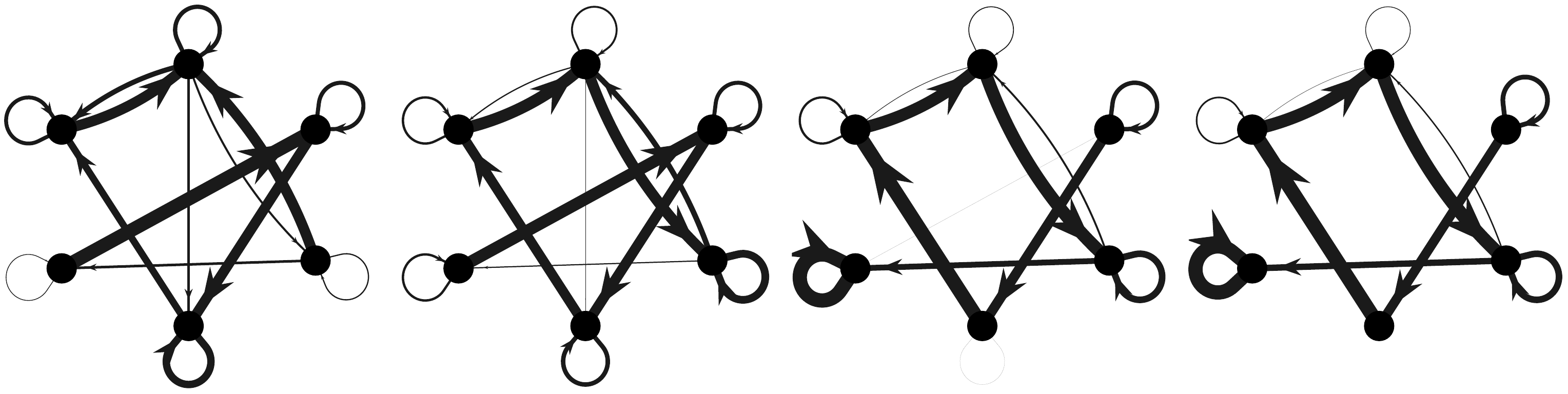}
			\end{tabular}
		\end{tabular}
		\caption{Visualization of the evolution of $\bm w(t)$ and $A(t)$ obeying the ASAP model~\eqref{eq:dynamics} with average-appraisal work flow~\eqref{eq:work-degree}. $A_0$ is strongly connected and $\workOpt$, $\bm w_0$, and $A_0$ satisfy the sufficient condition for failure to learn the optimal workload given by Lemma~\ref{thm:degree-centrality-necessary}. 
%			$\workOpt = [0.28, 0.02, 0.27, 0.01, 0.20, 0.22]^{\top} \leq \frac{1}{n}\ceil{A_{0}^\top}\ones[n]= [\frac{1}{3}, \frac{1}{2}, \frac{1}{3}, \frac{1}{3}, \frac{1}{2}, \frac{1}{2}]^{\top}$.
		}
		\label{fig:ex-neccCond}
	\end{subfigure}
	\caption{Various examples of cases where the team is unable to learn the work assignment that maximizes the collective team performance. For the work assignment vector, the darker the entry, the higher value it has. For the appraisal matrix, the thicker the edge is, the higher the appraisal edge weight is. 
%	For both subfigures, 
%%	$\workOpt = [0.28, 0.02, 0.27, 0.01, 0.20, 0.22]^{\top}$ and the
%%	initial work assignment $\bm w_0$; 
%	performance functions are of the form $p_i(w_i)=(\workOpti_i/w_i)^{\gamma}$ for $\gamma\in(0,1)$ and all $i$. 
	The plots pictured are at times $t=\{0,1,10,1000\}$, from left to right. 
%	Darker colors correspond to higher values of $\bm w(t)$ and $\workOpt$, and white corresponds to zero. Thicker edges in the appraisal networks correspond to larger appraisal weights.
}
%	Subfigures~\ref{fig:ex-eigvec} and \ref{fig:ex-neccCond} have the some initial appraisal matrix $A_0$. Subfigure~\ref{fig:ex-fail-reducible-eigvec} has one edge removed to make it weakly connected.
	\label{fig:examples-fail}
\end{figure}

\subsubsection*{Average-appraisal feedback limits direct cooperation}
\label{sec:fail2learn-degree}

Figure~\ref{fig:ex-neccCond} is an example of a team obeying the ASAP model~\eqref{eq:dynamics} with average-appraisal work flow~\eqref{eq:work-degree}. Even if the team does not satisfy the sufficient conditions for failure from Lemma~\ref{thm:degree-centrality-necessary}, when individuals adjust their work assignment with only their average-appraisal as the input, the team may still not succeed in learning the correct workload to maximize the team performance.

\section{Conclusion}
This paper proposes novel models for the evolution of interpersonal
appraisals and the assignment of workload in a team of individuals engaged
in a sequence of tasks. We propose appraisal networks as a mathematical
multi-agent model for the applied psychological concept of TMS.  For two
natural models of workload assignment, we establish conditions under which
a correct TMS develops and allows the team to achieve optimal workload
assignment and optimal performance.  Our two proposed workload assignment
mechanisms feature different degrees of coordination among team
members. The donor-controlled work flow model requires a higher level of
coordination compared to the average-appraisal work flow and, as a result,
achieves optimal behavior under weaker requirements on the initial appraisal
matrix.

%% On the other hand, the model with average-appraisal work flow
%% requires a lower level of interaction between team members but also
%% requires centralized information.

%% Then in order for teams obeying the average-appraisal work flow to achieve
%% optimal performance, the initial appraisal network requires additional
%% assumptions in addition to strong connectivity.

Possible future research directions include studying team's behavior when
individuals in the team update their appraisals and work assignments
asynchronously. The updates could be modeled using an additional contact
network with switching topology.  More investigation can also be done to
determine if it is possible to predict which appraisal weights in a weakly
connected network approach zero asymptotically, using only information on
the initial work distribution and appraisal values.

\section{Code Availability}
The source code is publicly available under {\normalsize\url{https://github.com/eyhuang66/assign-appraise-dynamics-of-teams}}.\hfill

\bibliographystyle{plainurl+isbn}
\bibliography{alias.bib,Main,FB}

\begin{IEEEbiography}
[{\includegraphics[width=1in,height=1.25in,clip,keepaspectratio]{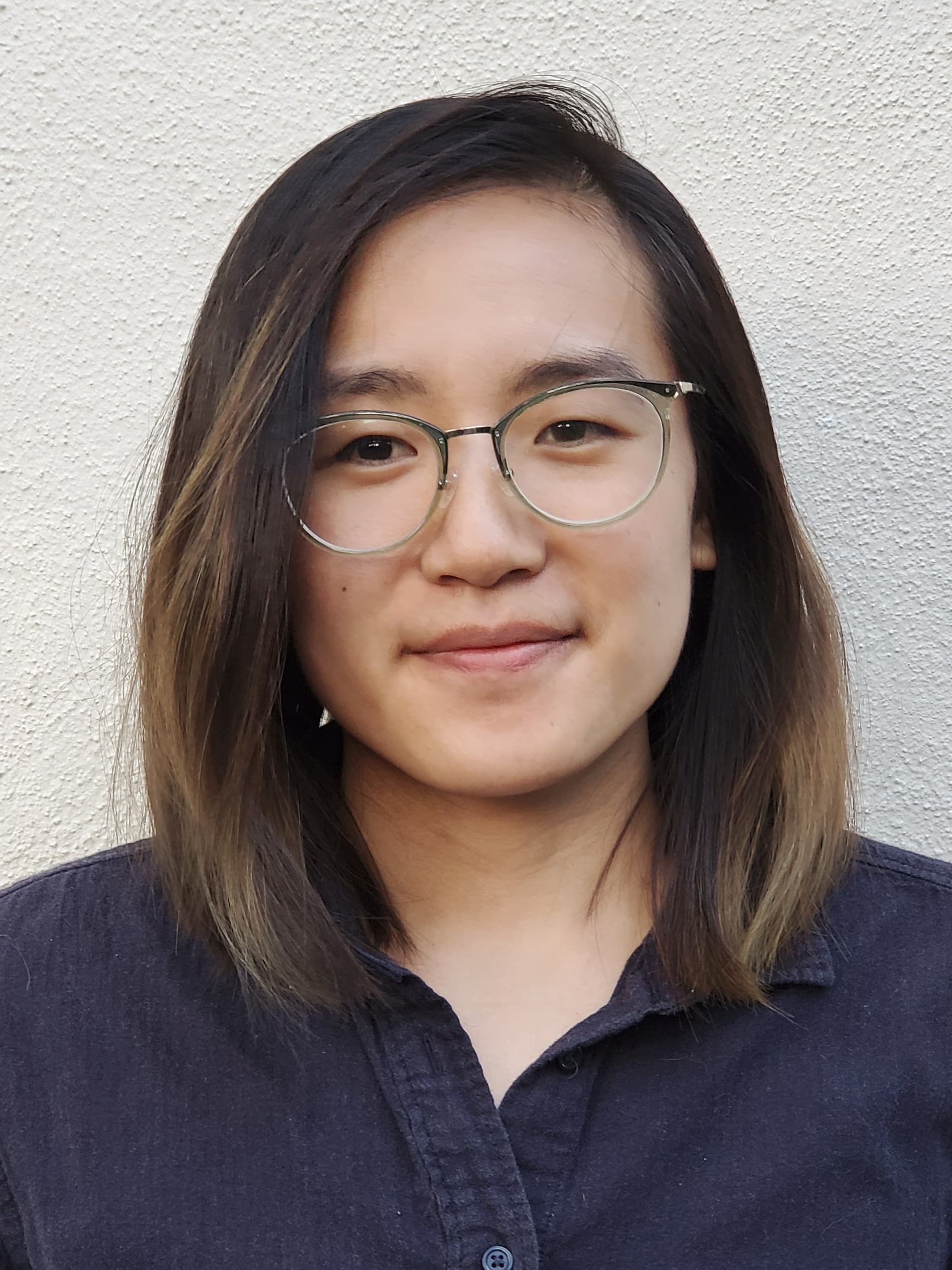}}]
{Elizabeth Y. Huang} received the B.S. degree in mechanical engineering from the University of California, San Diego, USA in 2016. She is currently working toward her Ph.D. in mechanical engineering from the University of California, Santa Barbara, USA. Her research interests include the application of control and algebraic graph theoretical tools for the study of networks of multi-agent systems, such as social networks, evolutionary dynamics, and power systems.
\end{IEEEbiography}

\begin{IEEEbiography}
	[{\includegraphics[width=1in,height=1.25in,clip,keepaspectratio]{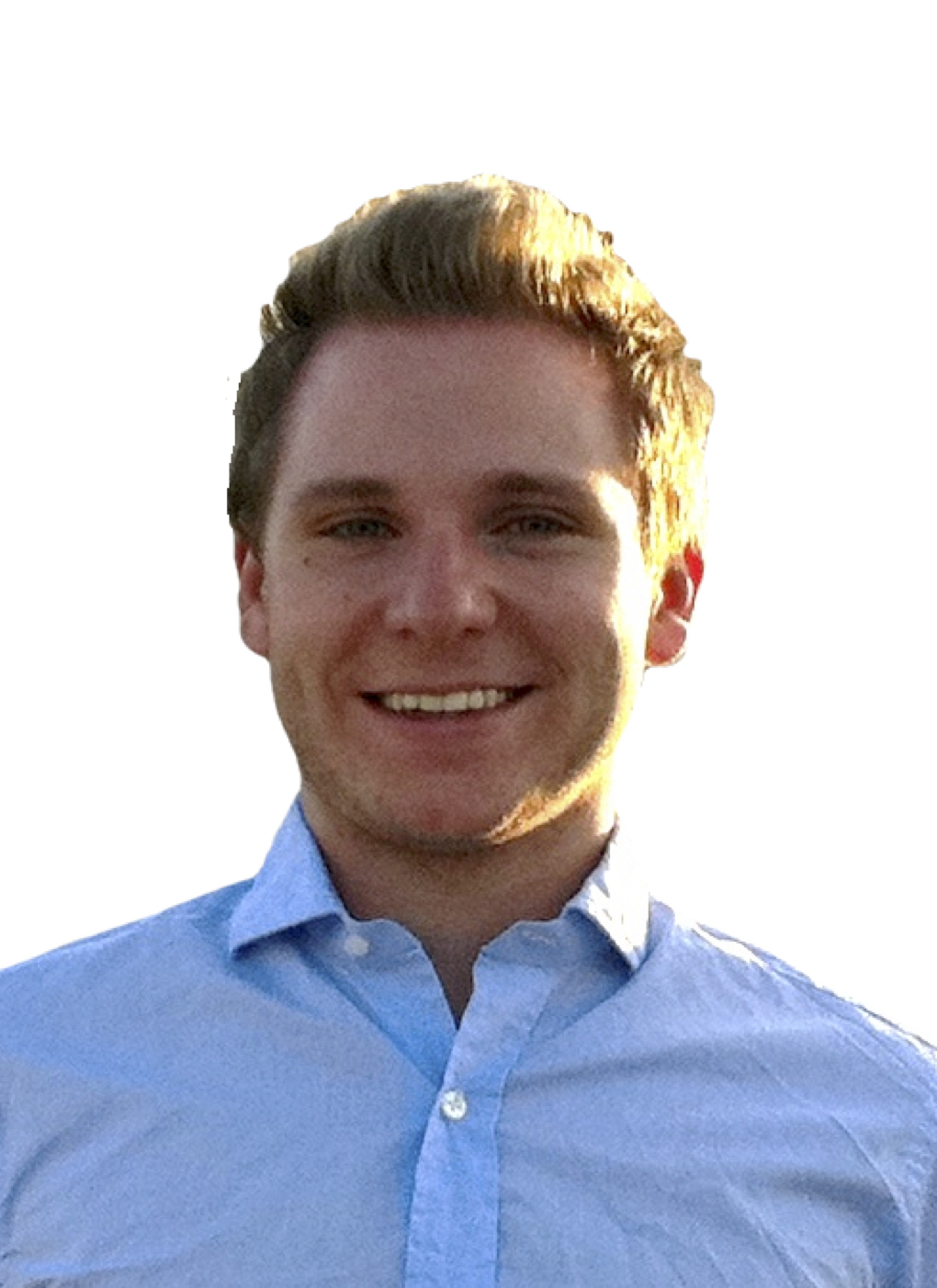}}]
	{Dario Paccagnan} is a Postdoctoral Fellow with the Mechanical Engineering Department and the Center for Control, Dynamical Systems and Computation, University of California, Santa Barbara.
	In 2018, Dario obtained a Ph.D. degree from the Information Technology and Electrical Engineering Department, ETH Z\"{u}rich, Switzerland. He received a B.Sc. and M.Sc. in Aerospace Engineering in 2011 and 2014 from the University of Padova, Italy, and a M.Sc. in Mathematical Modelling and Computation from the Technical University of Denmark in 2014; all with Honors.
	Dario was a visiting scholar at the University of California, Santa Barbara in 2017, and at Imperial College of London, in 2014.
	His interests are at the interface between control theory and game theory, with a focus on the design of behavior-influencing mechanisms for socio-technical systems. Applications include multiagent systems and smart cities. Dr. Paccagnan was awarded the ETH medal, and is recipient of the SNSF fellowship for his work in Distributed Optimization and Game Design. 
\end{IEEEbiography}

\begin{IEEEbiography}
[{\includegraphics[width=1in,height=1.25in,clip,keepaspectratio]{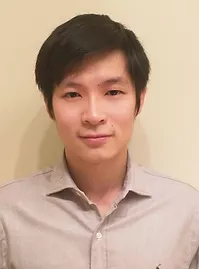}}]
{Wenjun Mei} is a postdoctoral researcher in the Automatic Control Laboratory at ETH, Zurich. He received the Bachelor of Science degree in Theoretical and Applied Mechanics from Peking University in 2011 and the Ph.D degree in Mechanical Engineering from University of California,
Santa Barbara, in 2017. He is on the editorial board of the Journal of Mathematical Sociology. His current research interests focus on network multi-agent systems, including social, economic and engineering networks, population games and evolutionary dynamics, network games and optimization.
\end{IEEEbiography}

\begin{IEEEbiography}
%[{\includegraphics[height=1.3in]{figures/Author_FrancescoBullo}}]
[{\includegraphics[width=1in,height=1.25in,clip,keepaspectratio]{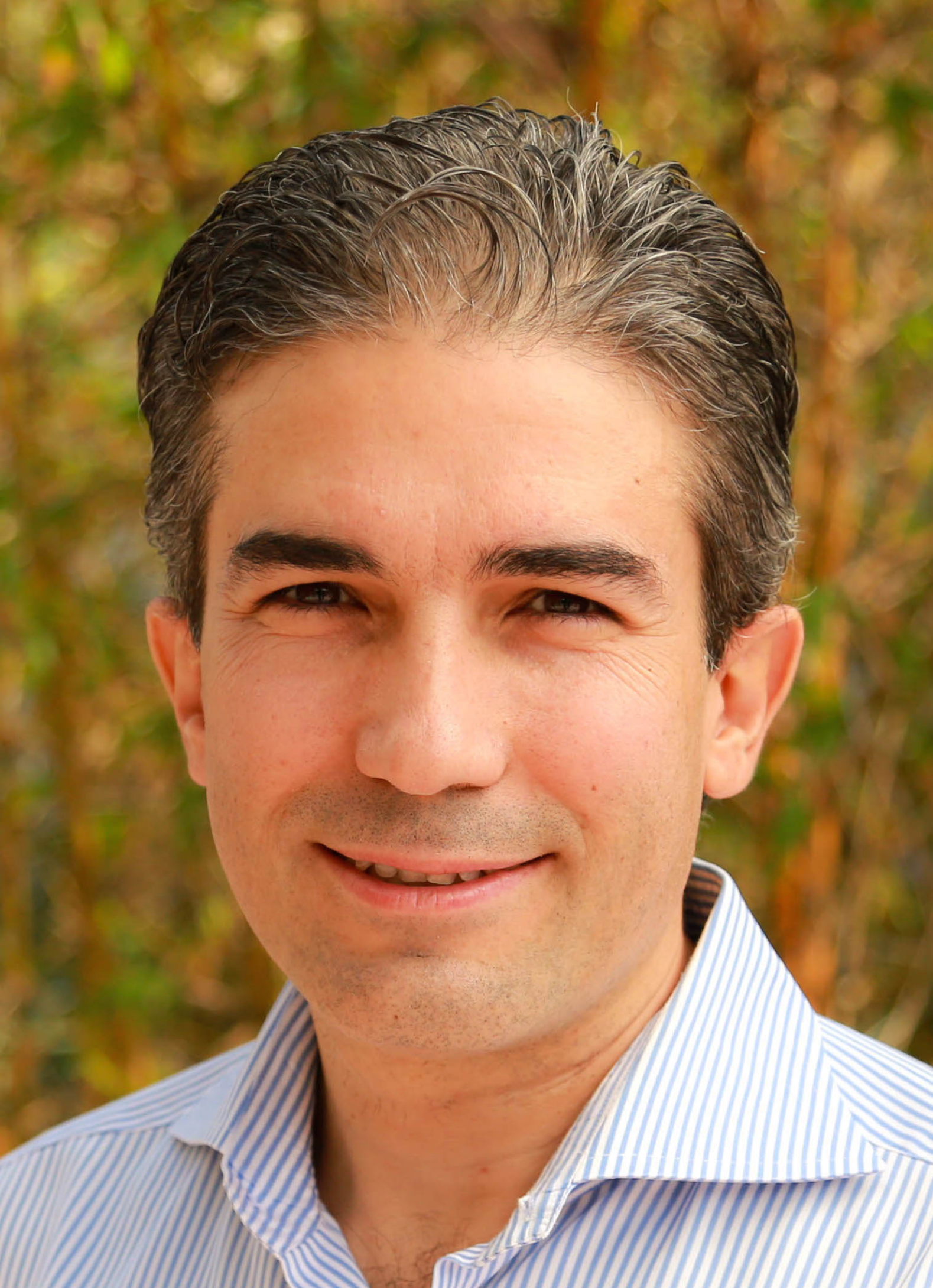}}]
{Francesco Bullo}(S'95-M'99-SM'03-F'10)
 is a Professor with the Mechanical Engineering Department and the Center
 for Control, Dynamical Systems and Computation at the University of
 California, Santa Barbara. He was previously associated with the
 University of Padova (Laurea degree in Electrical Engineering, 1994), the
 California Institute of Technology (Ph.D. degree in Control and Dynamical
 Systems, 1999), and the University of Illinois. He served on the editorial
 boards of IEEE, SIAM, and ESAIM journals and as IEEE CSS President. His
 research interests focus on network systems and distributed control with
 application to robotic coordination, power grids and social networks. He
 is the coauthor of “Geometric Control of Mechanical Systems” (Springer,
 2004), “Distributed Control of Robotic Networks” (Princeton, 2009), and
 “Lectures on Network Systems” (Kindle Direct Publishing, 2019, v1.3). He
 received best paper awards for his work in IEEE Control Systems,
 Automatica, SIAM Journal on Control and Optimization, IEEE Transactions on
 Circuits and Systems, and IEEE Transactions on Control of Network
 Systems. He is a Fellow of IEEE, IFAC, and SIAM.  
\end{IEEEbiography}

\end{document}